\documentclass[conference]{IEEEtran}

\AtBeginDocument{%
  \providecommand\BibTeX{{%
    \normalfont B\kern-0.5em{\scshape i\kern-0.25em b}\kern-0.8em\TeX}}}

\usepackage{amsmath,amsthm,amsfonts}
\usepackage{algorithmic}
\usepackage[ruled, lined, vlined, linesnumbered]{algorithm2e}
\usepackage{array}
\usepackage{booktabs}
\usepackage{epstopdf}
\usepackage{rotating}
\usepackage{tabularx}
\usepackage{threeparttable}
\usepackage{multirow}
\usepackage{multicol}
\usepackage{nccmath}
\usepackage{subcaption}
\usepackage{stfloats}
\usepackage{xcolor}
\usepackage{xspace}

\newtheorem{theorem}{Theorem}

\newtheorem{definition}{Definition}

\newtheorem{corollary}{Corollary}

\newcommand{\etal}{\emph{et~al.}\xspace}
\newcommand{\eg}{\emph{e.g.},\xspace}
\newcommand{\ie}{\emph{i.e.},\xspace}
\newcommand{\etc}{etc.\xspace}

\newcommand\figref[1]{\figurename~\ref{#1}}

\newcommand\tabref[1]{Table~\ref{#1}}

\newcommand\secref[1]{Sec.~\ref{#1}}
\newcommand\equref[1]{Eq.~(\ref{#1})}

\newcommand\algref[1]{Alg.~\ref{#1}}

\newcommand\thmref[1]{Theorem~\ref{#1}}
\newcommand\corref[1]{Corollary~\ref{#1}}

\newcommand{\fakeparagraph}[1]{\vspace{1mm}\noindent\textbf{#1.}}

\ifodd 0

\newcommand{\shuyue}[1]{{\color{blue}#1}}

\else

\newcommand{\shuyue}[1]{#1}

\fi

\ifCLASSOPTIONcompsoc
  \usepackage[nocompress]{cite}
\else
  \usepackage{cite}
\fi

\begin{document}

\title{Towards Capacity-Aware Broker Matching: \\ From Recommendation to Assignment}

\author{
    \IEEEauthorblockN{ 
        Shuyue Wei$^{1}$, Yongxin Tong$^{1}$, Zimu Zhou$^{2}$, Qiaoyang Liu$^{1}$, Lulu Zhang$^{3}$, Yuxiang Zeng$^{1}$, Jieping Ye$^{3}$
        }
        
    \IEEEauthorblockA{ 
        $^{1}$ State Key Laboratory of Software Development Environment, Beihang University, Beijing, China\\
        $^{2}$ City University of Hong Kong, Hong Kong, China \quad
        $^{3}$ Ke Holdings Inc., Beijing, China \\
        {$^{1}$ \{weishuyue, yxtong, qiaoyangliu, turf1013\}@buaa.edu.cn, 
        $^{2}$ zimuzhou@cityu.edu.hk} \\
        {$^{3}$ \{zhanglulu026, yejieping\}@ke.com}
    }
    
}

\markboth{Journal of \LaTeX\ Class Files,~Vol.~14, No.~8, August~2015}%
{Shell \MakeLowercase{\textit{et al.}}: Bare Demo of IEEEtran.cls for Computer Society Journals}

\IEEEdisplaynontitleabstractindextext

\maketitle

\begin{abstract}\label{sec:abstract}
Online real estate platforms are gaining increasing popularity, where a central issue is to match brokers with clients for potential housing transactions.
Mainstream platforms match brokers via top-\textit{k} recommendation.
Yet we observe through extensive data analysis that such top-\textit{k} recommendation tends to overload the top brokers, which notably degrades their service quality.
In this paper, we propose to avoid such overloading in broker matching via the paradigm shift from recommendation to assignment.
To this end, we design learned assignment with {contextual} bandits ({LACB}), a data-driven capacity-aware assignment scheme for broker matching which estimates broker-specific workload capacity in an online fashion and assigns brokers to clients from a global perspective to maximize the overall service quality.
Extensive evaluations on {synthetic and} real-world datasets from an industrial online real estate platform validate the {efficiency} and effectiveness of our solution.
\end{abstract}


\section{Introduction} \label{sec:intro}

Online real estate platforms, such as Compass\footnote{https://www.compass.com/}, Zillow\footnote{https://www.zillow.com} and Ke Holdings Inc. (a.k.a Beike)\footnote{https://www.ke.com} are increasingly exploiting data-driven approaches to improve their business and service quality.
A central function of these platforms is to match clients who are interested in house purchases to appropriate brokers\footnote{We do not differ an agent from a broker, and collectively refer them as brokers.} for followup services \cite{Book_real_estate_match}.
The status quo of such broker matching is top-\textit{k} recommendation \cite{Top-K, KDD20_tutorial}.
Take Beike, the largest Chinese online real estate platform, as an example. 
When clients click for detailed information about a house on the platform app, three brokers associated with that house will be recommended by the App (see \figref{fig:rec_unfair}).
    
\begin{figure}[t]
    \centering
    \includegraphics[width=0.365\textwidth]{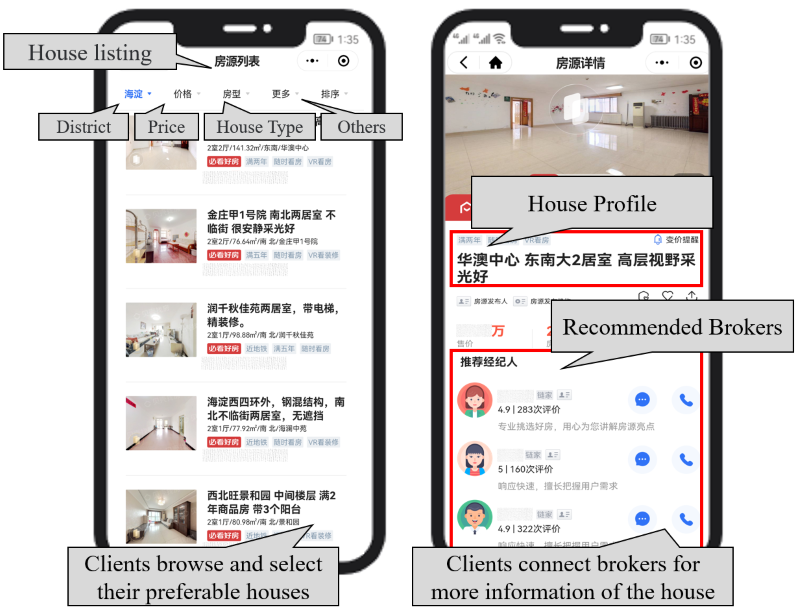}
    \caption{
    {\small 
    An illustration of top-\textit{k} broker recommendation on Beike, a major online real estate platform.
    When house buyers click for more information about their favorite houses, the platform App recommends three brokers by default.
    }
    } 
    \label{fig:rec_unfair}
    \vspace{-1.8em}
\end{figure}

We observe, through extensive data analysis on online real-estate platforms, that \textit{(i)} brokers have limited workload capacity and \textit{(ii)} the top-\textit{k} recommendation mechanism leads to the \textit{overloaded of top brokers}, also known as the \textit{overloaded phenomenon}, which impairs both the service quality and the long-term development of the platform.
Specifically, our study shows that due to the top-\textit{k} mechanism, the brokers' sign-up rate can drop from 14.3\%$\sim$27.5\% to 2.5\%$\sim$17.8\%, if they respond to over 40 client requests per day (see \secref{subsec:overloaded}).
Here the sign-up rate of a broker is a common indicator of service quality, which is the ratio between the number of clients who sign up with him/her and the total clients he/she served.
We also observe the Matthew effect \cite{Matthew} when adopting the top-\textit{k} recommendation mechanism.
That is, many requests are occupied by top brokers, leaving others few opportunities to improve their skills.
It may discourage those neglected brokers and harm the platform in the long run.

We argue that the overloaded phenomenon is caused by ignorance of brokers' workload capacity, which motivates us to take an \textit{assignment} \cite{ICDE20_task_assignment, SIGMOD15_task_assignment, ICDE16_online_task} perspective for capacity-aware broker matching.
That is, rather than blindly recommend the few top brokers to all clients, we propose to first estimate the workload capacity of individual brokers, and then assign them to clients in a global view without overwhelming the brokers.
However, it faces two practical challenges to implement capacity-aware assignment for broker matching.
\begin{itemize}
    \item 
    \textit{Challenge 1: how to estimate broker-specific workload capacity
    {in an online fashion}?}
    We observe that the workload capacity differs across brokers (see \secref{sec:motivation}), making personalized estimation necessary.
    However, it is impractical to collect data on a broker's service quality under all possible workloads 
    {in advance, which makes online estimation of workload capacity preferable}.
    Prior workload capacity estimation schemes \cite{WWW10_bandit_rec, ICML20_nn_bandit} fail to support 
    {such online learning of personalized estimation}.
    \item 
    \textit{Challenge 2: how to assign brokers under capacity limit to maximize the overall utility over time?}
    It is common that the amount of real estate transactions at present affects that in the near future.
    Consequently, broker assignments among batches tend to be correlated, making it difficult to assign brokers holistically.
    Most assignment schemes \cite{ICDE19_batch, KDD17_dispatch} consider clients and brokers in each batch independently, making them sub-optimal in terms of the collective utility of multiple batches.
\end{itemize}

To address these challenges, we propose \underline{L}earned \underline{A}ssignment with \underline{C}ontextual \underline{B}andits (LACB), a {data-driven} capacity-aware assignment scheme for real estate broker matching.
It tackles \textit{Challenge 1} via contextual bandits for data-efficient and {online} personalized capacity estimation.
LACB overcomes \textit{Challenge 2} with a capacity-aware value function, which accounts for both the short-term and the long-term utility of brokers for matching.
Our main contributions and results are summarized as follows.
\begin{itemize}
    \item 
    To the best of our knowledge, we are the first to identify the overload of top brokers problem for online real estate platforms.
    Extensive data analysis shows that brokers have limited workload capacity and their service quality tends to drop when overloaded, which motivates the shift from recommendation to assignment for broker matching.
    \item 
    We design LACB, a data-driven capacity-aware assignment scheme for broker matching.
    It estimates broker-specific capacity in an online manner and assigns brokers to clients from a global view.
    We further propose LACB-Opt, which accelerates assignments via broker selection.
    \item 
    We conduct extensive experiments on synthetic and real-world datasets from Beike, the largest Chinese online real estate platform.
    Experimental results validate the efficiency and effectiveness of our solutions.
\end{itemize}

In the rest of this paper, we first identify the overloaded phenomenon in \secref{sec:motivation} and formulate the problem in \secref{sec:problem}.
Then we present an overview of our solution in \secref{sec:method}, and introduce each module in  \secref{sec:method_estimate} and \secref{sec:method_assignment}, respectively.
We present the evaluations in \secref{sec:experiment}, review related work in \secref{sec:related} and finally conclude in \secref{sec:conclusion}.

\section{Motivations} \label{sec:motivation}
We motivate our study through measurements from Beike, an online real estate platform in China.
We observe a phenomenon called \textit{the overload of top brokers}, where a few brokers are tasked to serve amounts of requests that exceed their capacities, which eventually leads to drop in the brokers' service quality and the overall utility of the platform.

\subsection{Limited Broker Capacity} 
\label{subsec:limited}
Our first motivation is that brokers have limited capacity.
As with other service industries, we assume a real estate broker has limited capacity, \ie the number of services he/she can provide in unit time with high quality.
Since the low quality of service in housing transactions easily leads to the churn of clients, we hypothesize that the service quality of brokers tends to drop with the increase of served requests.
We test this hypothesis via the measurements below.

\fakeparagraph{Measurements}
We analyzed data from an online real estate platform from June 1st to August 31st, 2021 in two major cities in China to explore the relationship between service quality and capacity of brokers.
We use the broker's sign-up rate, \ie the ratio between the number of clients who sign up with the broker and that served in total, as the proxy for the service quality.
We measure the sign-up rates as the increase of workload, \ie the number of requests served daily, at both the city and individual levels.

\begin{figure}[t]
    \centering
    \includegraphics[width=0.35\textwidth]{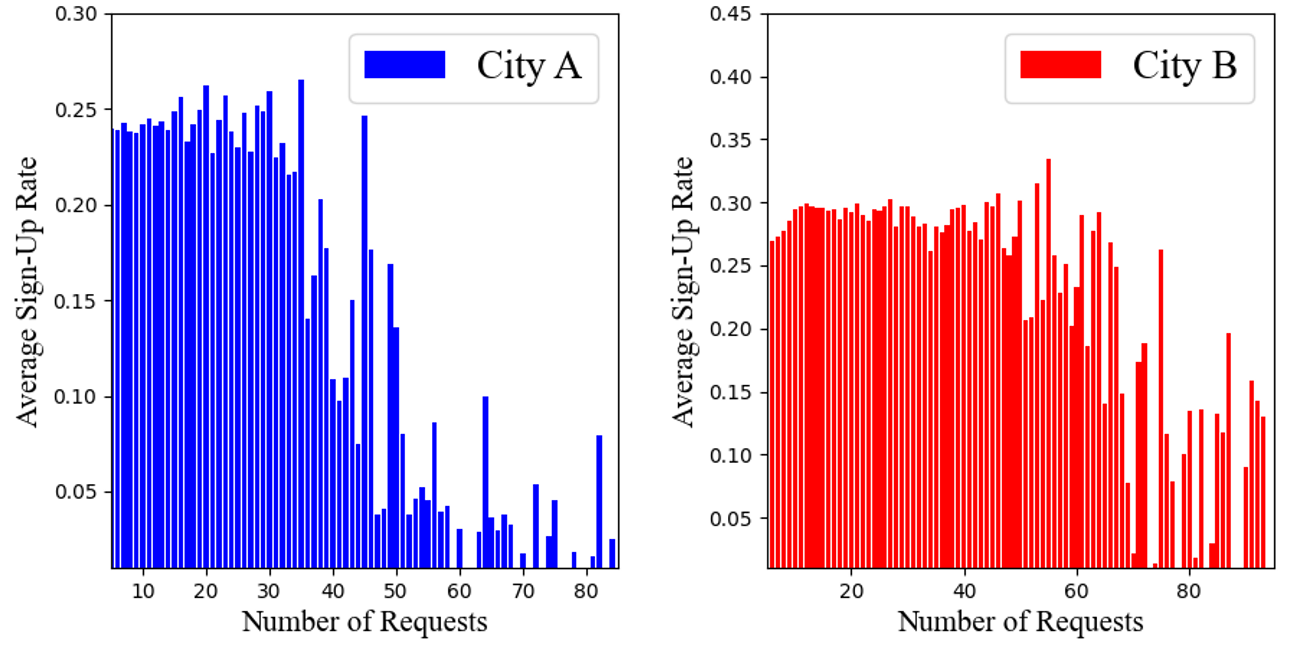}
    \caption{
    {\small The average sign-up rate of brokers in two cities. 
    }
    }
    \label{fig:overall_agents}
     \vspace{-1.8em}
\end{figure}

\begin{figure*}[t]
    \centering
    \includegraphics[width=0.9\linewidth]{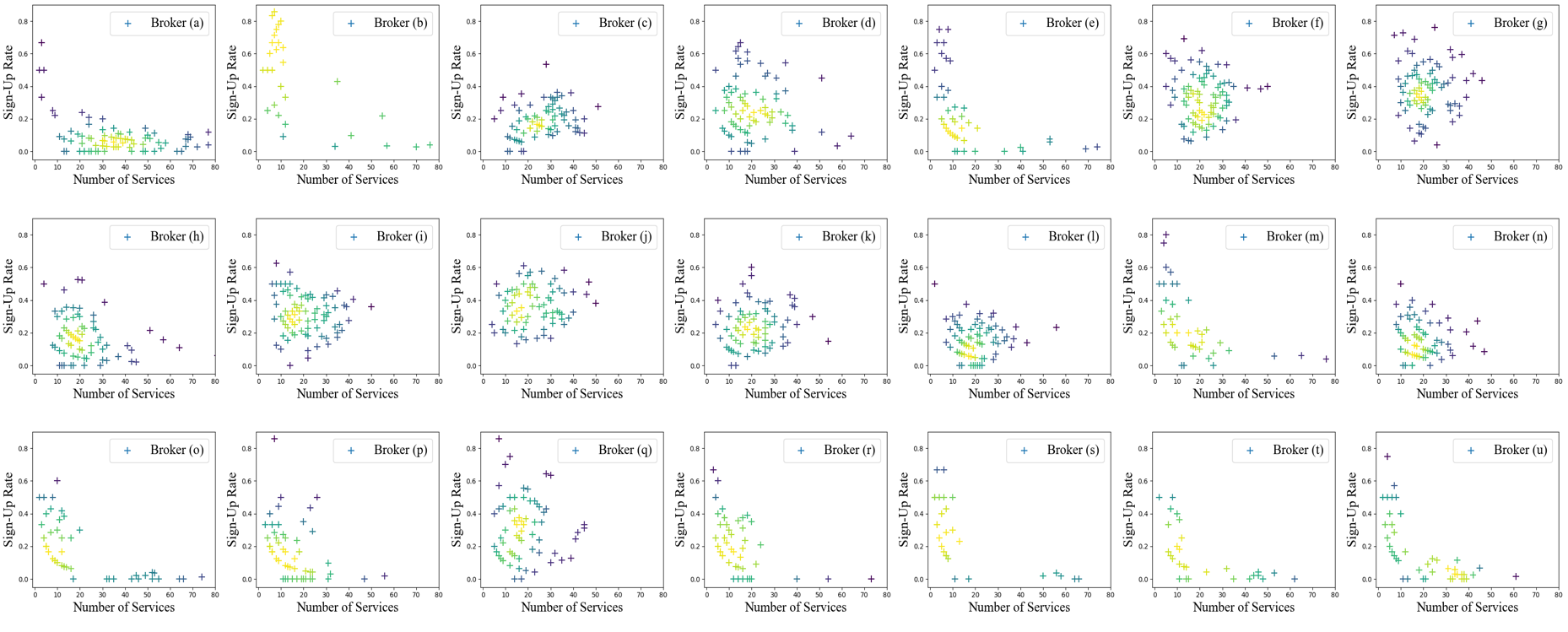}
    \vspace{-0.8em}
    \caption{
    {\small The average sign-up rate of top brokers in City A. 
    We apply the Gaussian kernel density estimation to fit the empirical distribution of broker's performance with workload.
    The center of the performance distribution is in the lighter color, which represents their accustomed workload area.
    Most top brokers perform better in light area compared with large workload area.
    }
    }
    \label{fig:agents}
    \vspace{-1.5em}
\end{figure*}

\fakeparagraph{Observations}
We observe that the sign-up rates tend to drop with the increase of workload, 
and the decreasing patterns seem complex and broker-specific.
\begin{itemize}
    \item 
    \figref{fig:overall_agents} shows the average sign-up rate of brokers in the two cities with the increase of requests served per day.
    Take City A (in blue) as the example. 
    When the number of requests served is below 40 per day, the average sign-up rate is 14.3$\sim$27.5\%.
    Once brokers have to tackle more than 40 requests a day, their average sign-up rate drops to 2.5$\sim$17.8\%. 
    By employing Welch's t-test, we find that the sign-up rate is statistically significantly correlated to the number of requests served daily ($p$-value $<$ 0.0001).
    Thus excessive workload of brokers lowers the service quality, and even leads to client churn.
    Similar decreasing trends are observed for City B (in red) as well.
    \item 
    We further study the top 50 brokers who serve most requests in City A, where 21 of them serve more than 40 requests occasionally.
    \figref{fig:agents} plots the sign-up rates of these 21 brokers with higher workloads in City A.
    Among all the 21 brokers, their sign-up rates exhibit a decreasing trend as the number of requests served daily increases.
    \item
    Despite the decreasing trend, the relationship between the sign-up rate and the number of requests served tends to be complex, non-linear and broker-specific patterns, as observed from both \figref{fig:overall_agents} and \figref{fig:agents}.
\end{itemize}
    
\subsection{Overloaded Top Brokers}
\label{subsec:overloaded}
Our second motivation is that the top brokers tend to be overloaded due to the top-\textit{k} recommendation mechanism in current online real estate platforms.
This is because the platform lists the top-\textit{k} brokers without accounting for their capacities, while clients incline to select from the top brokers listed by the platform.
We test this claim as follows.

\fakeparagraph{Measurements}
We analyze data of the same online real estate platform in June, 2021 in City A and plot the workload distribution breakdown of brokers recommended by the platform and those not listed by the platform.
By default, the platform recommends the top-3 brokers (see \figref{fig:rec_unfair}).

\begin{figure}[t]
      \centering
      \includegraphics[width=0.35\textwidth]{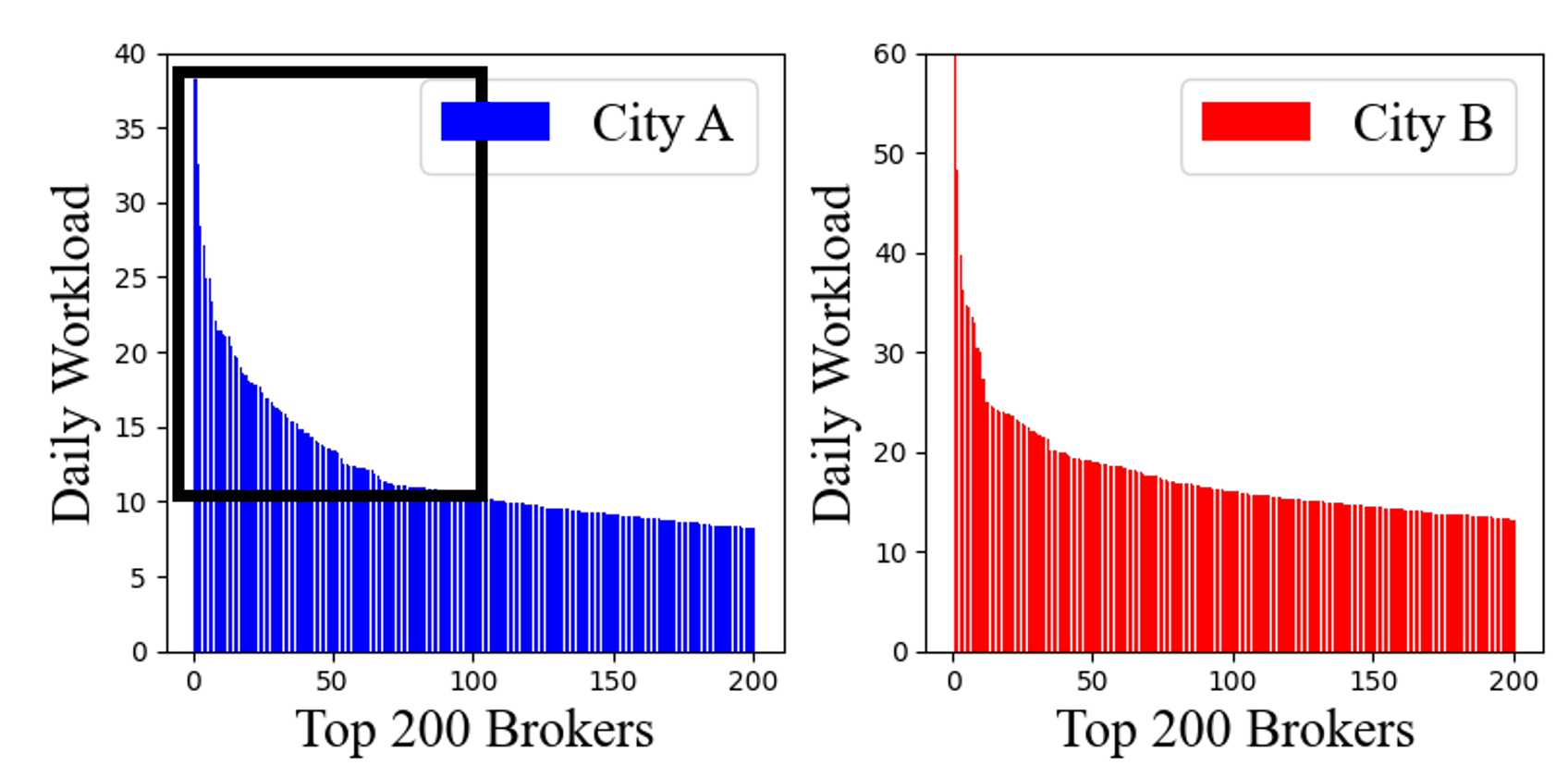}
      \caption{
      {\small 
      The workload distribution of the top brokers in City A and City B.
      People tend to follow the recommended brokers and most requests are served by these top brokers. 
      }
      }
      \label{fig:workload_number}
      \vspace{-2em}
\end{figure}   

\fakeparagraph{Observations}
We observe that the workload distribution is highly unbalanced towards the top-3 brokers recommended by the platform.
In City A, the top-1 broker serves 38.26 requests daily on average, while a broker serves 3.18 requests per day on average, \ie the top-1 broker's workload is 12.03$\times$ larger than the average workload.
\figref{fig:workload_number} plots the workload distribution of top-200 brokers in both City A and City B.
As is shown, their workloads are all notably higher than the city average workload.
Note that from \figref{fig:agents}, a workload of 10 to 20 results in higher sign-up rates.
Hence, roughly a hundred brokers in the black box risk exceeding their limited capacity.

Furthermore, the \textit{Matthew effect} may occur if top brokers are continually tasked requests.
That is, most requests are occupied by top brokers, while others are overlooked.
As a result, the neglected brokers have few opportunities to improve their home-finding skills, which has a negative impact on the development of the platform.
Overloaded top brokers are also common in City B.  

\subsection{Key Observations and Insights}
\label{subsec:summary}

In summary, we observe that the top-\textit{k} recommendation mechanism used by prior online real estate platforms tends to overload the top brokers, which we call the \textit{overload of top brokers} problem.
Overloaded brokers exhibit drop in service quality, eventually leading to a decrease in sign-up rates.
The overloaded problem occurs because the top-\textit{k} recommendation ignores the broker's capacity.

The overloaded phenomenon motivates us to rethink broker matching from an \textit{assignment} perspective.
Rather than blindly recommend a small group of top brokers to all clients, we propose to assign brokers to clients from a global view, while accounting for the capacity of top brokers.
As next, we formulate our perspectives into a capacity-aware assignment problem and propose practical solutions for efficient broker matching in case of unknown workload capacity.

\section{Problem Statement}\label{sec:problem}
In this section, we formally define the capacity-aware assignment (CAA) problem to avoid overloading top brokers during broker matching.
\begin{definition}[Broker] \label{def:broker} 
A broker $b\in B$ is a triple $(\textbf{x}_b, w_b, {s}_{b})$, where $\textbf{x}_b$, $w_b$, and ${s}_{b}$ represent broker $b$'s attributes, the number of requests served daily, and the daily sign-up rate.
\end{definition}
The vector $\textbf{x}_b$ includes attributes such as working years, job title, average dialogue rounds, \etc (details in \tabref{tlb:feature}), which reflects the current working status of broker $b$.
Aligned with our measurements in \secref{sec:motivation}, we use the number of requests served daily $w_b$ to quantify a broker's workload, and the daily sign-up rate ${s}_b$ as the proxy of his/her service quality.
For ease of presentation, we use $\mathcal{T}_b$ to represent the trial triples of broker $b$, \ie $\mathcal{T}_b=\{(\textbf{x}_{1}, w_1, {s}_{1}),...,(\textbf{x}_{d}, w_{d}, {s}_{d})\}$.

In broker matching, brokers are tasked to serve clients' requests. 
Following the conventions in the real estate industry, we perform broker matching on a batch basis and adopt the sign-up rate as the utility of assignment.
\shuyue{We take the fixed-time window batched assignment, \ie the platform presets the time interval and assigns brokers to all the requests appeared.}
Denote the requests appearing in time interval $i\in I$ as $R^{(i)}$.
From the assignment perspective, broker matching can be formulated as the following problem.
\begin{definition}[\underline{C}apacity-\underline{A}ware \underline{A}ssignment (CAA) Problem] \label{def:problem}
\shuyue{
Given a set of brokers $B$, requests in each interval $R^{(i)}$,
and the matching utility between them $u_{r,b}$ (utility of assigning broker $b$ to request $r$), we aim to find assignments $\mathcal{M}^{(i)}$ for each time interval $i$ that maximizes the total utility under the unknown capacity constraint.}
\begin{itemize}
    \item Maximizing Total Utility.
    \begin{equation}
    {\max \sum_{i\in I}\sum_{r,b} u_{r,b}\mathcal{M}_{r,b}^{(i)}}
    \end{equation}
    \item Capacity Constraint.
    \begin{equation}
    {\forall b,\quad \sum_{i\in I}\sum_{r} \mathcal{M}_{r,b}^{(i)} \leq {{c}_{b}}}
    \end{equation}
\end{itemize}
\end{definition}
where $\mathcal{M}_{r,b}^{(i)}=1$ if broker $b$ is assigned to request $r$ and $\mathcal{M}_{r,b}^{(i)}=0$ otherwise.
$u_{r,b}$ is the utility obtained if broker $b$ is assigned to request $r$, \shuyue{which is the input and can be learned from historical assignments using models such as XGboost \cite{KDD16_Xgboost}}.
\shuyue{${c}_{b}$ is the unknown capacity of broker $b$ to be estimated.}
\shuyue{\tabref{tlb:notions} summarizes the major notions throughout this paper.}

\fakeparagraph{Discussions}
We make 
{two} notes on the CAA problem.
\begin{itemize}
    \item 
    Although the batched assignment modeling has been widely adopted in 
    crowdsourcing applications \eg
    ride-hailing
    \cite{KDD18_batch, ICDE19_batch}, 
    it is the first time it is catered for broker matching for online real estate platforms.
    \item
    A unique challenge of the CAA problem against the general batched assignment lies in the broker capacity ${c}_{b}$, which is not given in advance.
    Thus, prior to assignment, effective capacity estimation is necessary.
    Consequently, we aim to devise both capacity estimation and assignment algorithms tailored for online real estate applications.
    \item \shuyue{    Other well-known applications include online healthcare (\eg ZhongAn Insurance\footnote{https://www.zhongan.com/}) and legal consultation (\eg Lvshiguan\footnote{http://www.lvshiguan.com}), where a similar capacity-aware matching problem exists due to the limited workload capacity of workers to serve their assigned requests.}

\begin{table}[htb]
    \centering
    \caption{Summary of major symbols.}
    \label{tlb:notions}
    \begin{tabular}{ll}
    \hline
    \textbf{Notation}          & \textbf{Description}                               \\ \hline
    $\textbf{x}_b,w_b,s_b, c_b$ &attributes, workload, sign-up rate and capacity of broker $b$ \\
    $\mathcal{T}_b$           & trial triples of broker $b$             \\
    $I,i$                         & time intervals horizons and time interval $i$         \\
    $B$,$B_{+}$                 & set of all brokers and set of brokers with residue capacity\\
    $R, R^{(i)}$      &set of all requests and set of requests in interval $i$ \\
    $u_{r,b}$                   & utility of assigning broker $b$ to request $r$  \\
    $\mathcal{M}^{(i)}, \mathcal{M}^{(i)}_{r,b}$    & assignment in interval $i$ and assignment indicator of $(r,b)$         \\
    $UCB_{\textbf{x},c}$               &upper confidence bound of capacity $c$ under context $\textbf{x}$      \\
    $\mathcal{B}_{\theta, \textbf{D}}$      & contextual bandit with learned parameters $\theta, \textbf{D}$  \\
    $\mathcal{S}_{\theta}, g_{\theta}$  & reward mapping function and its gradient \\
    $\mathcal{L(\theta)}$      & loss function of neural networks in contextual bandit \\
    $\mathcal{B}_{b}$      & exclusive contextual bandit of broker $b$  \\
    $\mathcal{V}(cr)$              & capacity-aware value function under  residue capacity $cr$ \\
    $\beta, \gamma, \delta$   &learning rate, discount factor and update threshold \\
    
    \hline
    \end{tabular}
    \vspace{-1em}
\end{table}

\end{itemize}

\begin{table*}[t]
    \caption{ The broker's attributes.}

    \label{tlb:feature}
    \begin{tabular}{l|l|l}
        \hline
        \textbf{Attribute Type}                  & \textbf{Attribute}                & \textbf{Description}                                                          \\ \hline
        \multicolumn{1}{c|}{}                 & Age                             & Broker's age.                                                                  \\
        \multirow{2}{*}{\textbf{Basic Info.}} & Working Year               & The working years as a broker.                                                \\
                                              & Education                       & Education background (\eg{undergraduate, master}).           \\
        \multicolumn{1}{c|}{}                 & Title                           & Job title (\eg{assistant, clerk, manager}).                  \\ \hline
        \multicolumn{1}{c|}{}                 & Response Rate                   & The rate of the broker's response to a request in one minute.                \\
        \textbf{}                             & Dialogue rounds             & The average dialogue rounds via the App in recent 7/14/30/90 days.       \\
        \multicolumn{1}{c|}{}                  & Number of Housing Presentation & The number of broker's presenting houses offline in recent 7/14/30/90 days. \\
        \multicolumn{1}{c|}{}                  & Number of Presentation via VR   & The number of broker's presenting houses via VR in recent 7/14/30/90 days.   \\
        \multicolumn{1}{c|}{}                  & Time of Presentation via VR     & The time of broker's presenting houses via VR in recent 7/14/30/90 days.     \\
                                              & Number of Consultation via Phone & The number of broker answering clients via phone in recent 7/14/30/90 days.    \\
        \multirow{2}{*}{\textbf{Work Profile}} & Time of Consultation via Phone  & The time of broker answering clients via phone in recent 7/14/30/90 days.           \\
                                              & Number of Consultation via App   & The number of broker answering clients via App in recent 7/14/30/90 days.  \\
                                              & Time of Consultation via App     & The time of broker answering clients via App in recent 7/14/30/90 days.    \\
                                              & Number of Maintained Houses     & The number of houses currently maintained by the broker.                      \\
                                              & Number of Served Clients       & The number of clients who are served by the broker in recent 7/14/40/90 days. \\
                                              & Number of Housing Transactions  & The number of housing transactions through the broker in recent 7/14/40/90 days. \\ \hline
        \multirow{2}{*}{\textbf{Preference}}  & Districts Information                 & Broker's preferable communities and area around POIs.                         \\
                                              & Housing Information                   & Broker's preferable price, area and type of houses.                          \\ \hline
        \end{tabular}
    \vspace{-2em}
    \end{table*}
 
\begin{figure}[thb]
    \centering
    \includegraphics[width=0.9\linewidth]{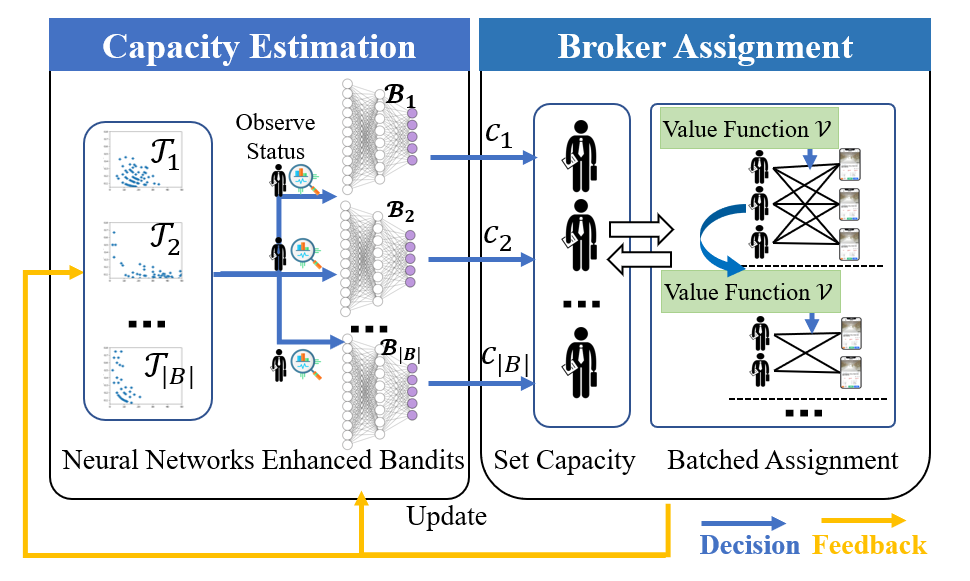}
    \caption{\small Workflow of LACB.}
    \label{fig:overview}
    \vspace{-2em}
\end{figure}

\section{Overview of Our Solution} \label{sec:method}
To solve the CAA problem, we propose \underline{L}earned \underline{A}ssignment with \underline{C}ontextual \underline{B}andits (LACB), which learns the unknown broker capacity via contextual bandit and assigns brokers from the global perspective to maximize the total utility without overloading the top brokers.
We first present an overview of LACB and explain each functional module. 

LACB consists of two functional modules, \textit{capacity estimation} and \textit{capacity-based assignment}.
\begin{itemize}
    \item 
    The \textit{capacity estimation} module decides the daily workload capacity according to the broker's current status by neural network enhanced contextual bandits.
    \item 
    The \textit{capacity-based assignment} module selects a set of brokers satisfying the capacity constraint and assigns them to requests via the capacity-aware value function.
\end{itemize}

\figref{fig:overview} shows the workflow of LACB. 
It operates in two phases: estimation and assignment.
First, we observe the broker's working status and set a  daily workload capacity for him/her by the neural network enhanced bandit.
In the assignment phase, we take the broker's estimated capacity and adopt a capacity value function to guide the assignment, capturing the long-term utility of brokers with different workloads.
Finally, we store the results of batched assignment as feedback to improve future decisions.

\section{Capacity Estimation \label{sec:method_estimate}}
This section introduces our capacity estimation method.
We formulate the workload capacity estimator as a contextual bandit and propose a neural network enhanced policy to decide the daily workload capacities for each broker.

\subsection{Basic Idea}
Our method is motivated by the following three challenges when estimating the broker capacities.
\begin{itemize}
    \item
    \textit{Online training}.
    It is impractical to collect data of a broker's sign-up rate under all possible workloads  prior to model deployment.
    The workload capacity estimator is expected to routinely adapt itself to the observations of workloads and sign-up rates after deployment.
    Our solution is to apply the contextual bandit algorithm \cite{bandit_book}, a learning method to explore the unknown environment and make decisions under various contextual information. 
    \item 
    \textit{Complex relations between a broker’s performance and workload}.
    The broker's performance is closely related to the current working status.
    For instance, a broker may be more exhausted in the sales seasons, and thus less resilient to heavy workloads.
    As observed in \secref{subsec:limited}, the relationship between a broker’s performance and his/her workload is non-linear, and we model such non-linear complexity via neural networks.
    \item 
    \textit{Personalized workload capacity}.
    As shown in \figref{fig:agents}, the relationship between sign-up rate and workload is not only complex but also broker-specific. 
    However, it is challenging to directly learn a personalized capacity estimator due to the sparsity of broker-specific data.
    We first learn a generic capacity estimation model and fine-tune it for individual brokers.
\end{itemize}
As next, we elaborate on the designs in sequel.

\subsection{Workload Capacity Estimator as Contextual Bandit}
As mentioned above, we utilize contextual bandits to learn a generic broker capacity estimator 
{in an online fashion by interacting with the real estate platform.}
\shuyue{The reinforcement learning \cite{RL_book} (\eg Q-learning) mainly models the effect of the decision on the state. 
However, in our scenario, the broker's intrinsic working status is not affected by our decisions, so approaches like Q-learning are infeasible to capacity estimation.}

\fakeparagraph{Review on Contextual Bandit} 
We start with a quick review of the contextual bandit.
A bandit with \textit{$k$-arms} is widely used for online decision-making in an unknown environment over $n$ 
{batches}, where each arm represents a decision.
In each batch, the bandit chooses one arm (decision) and receives a \textit{reward} from the environment.
It then updates its decision-making strategy based on the reward and tries to maximize the total rewards over the $n$ 
{batches}.
A contextual bandit further allows the bandit to make decisions with additional information (\ie the \textit{context}) at the beginning of each 
{batch}.

\fakeparagraph{Our Formulation} 
We now explain how to formulate the workload capacity estimator as a contextual bandit.
We consider the broker's candidate workload capacities as arms of the bandit (represented as $\mathcal{C}$).
The broker's working status $\textbf{x}_b$ is viewed as the \text{context}, based on which the bandit chooses a capacity $c_b\in\mathcal{C}$.
The daily sign-up rate ${s}_{b}$ under workload $w_b$ is used as the reward.
{The workload capacity estimator interacts with the real estate platform}, which is viewed as the unknown environment.
In each {batch}, the real estate platform executes assignment algorithms and reveals the reward ${s}_{b}$.
We use $(\textbf{x}_b, w_b, {s}_b )$ as a trial triple to update the reward function of the bandit (workload capacity estimator) since a broker's workload $w_b$ is usually lower than her/his capacity $c_b$.

\subsection{Choosing Capacity with Neural Network Enhanced UCB}
{With the workload capacity estimator formulated as a contextual bandit, the next question is to determine the \textit{policy} to choose the workload capacity that maximizes the daily sign-up rates for the given broker working status.}

\fakeparagraph{Standard UCB}
A common decision-making policy for a contextual bandit is the upper confidence bound (UCB) algorithm \cite{bandit_book}.
It acts as if the environment is as nice as plausibly possible and uses a context vector to calculate the upper confidence bound of the expected reward.
Then it chooses the arm with the maximum upper confidence bound as the decision.
For our workload capacity estimation, we can use the broker's working status $\textbf{x}$ as the context, and calculate the upper confidence bound $UCB_{\textbf{x},c}$ of each workload capacity $c$ using the equation below \cite{bandit_book}.
\begin{equation} \label{eq:ucb_linear}
    \medop{
    {UCB_{\textbf{x},c} = f_{\theta}(\textbf{x}, c)+\alpha\sqrt{f'_{\theta}(\textbf{x},c)^{T} D^{-1} f'_{\theta}(\textbf{x},c)}}
    }
\end{equation}

where $f_{\theta}(\textbf{x}, c)$ is a linear model which maps the context $\textbf{x}$ to the expected reward of a workload capacity and $f'_{\theta}(\textbf{x},c)$ is its derivative.
$\alpha$ is a preset coefficient parameter.
$\theta$ is the parameters of the linear model and $D$ is the covariance matrix.
$\theta$ and $D$ are parameters to be trained.

\fakeparagraph{NN-enhanced UCB}
A limitation of the standard UCB algorithm is the assumption on the linear relation between the expected reward and the context, \ie $f_{\theta}(\textbf{x}, c)$ in \equref{eq:ucb_linear}.
Hence, the standard UCB fails to depict the non-linear relation between the broker's sign-up rate (expect reward) and the working status (context) in our scenario (see \secref{subsec:limited}).
As a remedy, we replace the linear model by a neural network.
We name the corresponding capacity choosing policy as NN-enhanced UCB.

\begin{figure}[t]
    \centering
    \includegraphics[width=0.8\linewidth]{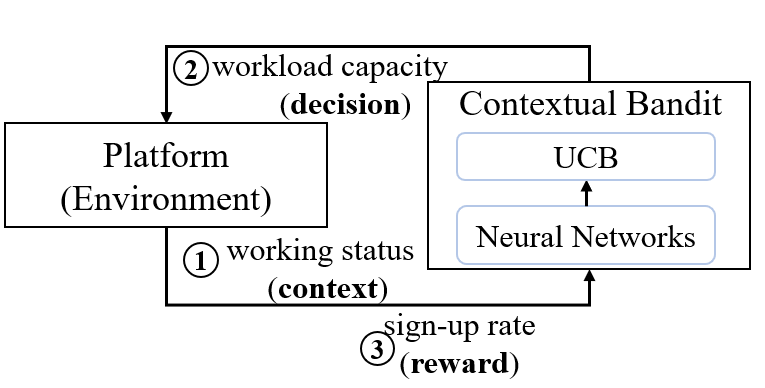}
    \caption{
    \small Workflow of the NN-enhanced UCB.
    } 
    \label{fig:bandit}
    \vspace{-2em}
\end{figure}

\figref{fig:bandit} shows the workflow of our NN-enhanced UCB.
We replace the linear model $f_{\theta}(\textbf{x}, c)$ in standard UCB by a learned reward mapping function $\mathcal{S}_{\theta}(\textbf{x}, c)$.
For simplicity, we adopt a fully connected MLP network with $L$ layers:
\begin{equation}\label{eq:nn}
    \medop
    {\mathcal{S}_{\theta}(\textbf{x},c)= {W}_{L}\cdot\sigma_{L-1} ( \cdots\sigma_{1}({W}_{1}[\textbf{x};c]))}
\end{equation}
where $W_i$ $(1\le i \le L)$ are the learned parameters of each layer and $W_1\in \mathbb{R}^{m\times d}$, $W_i \in \mathbb{R}^{m\times m}$ $(2\le i \le L-1)$, $W_L \in \mathbb{R}^{m\times 1}$.
$\sigma_{i}$ is the ReLU activation.

Let ${\theta}\in\mathbb{R}^{d}$ be all the learnable parameters of the neural network and its gradient is denoted by $g_{\theta}(\textbf{x},c)= \nabla_{\theta}\mathcal{S}_{\theta} \in \mathbb{R}^{d}$.
Then we can easily extend the upper confidence bound of the expected reward in \equref{eq:ucb_linear} as follows.
\begin{equation}\label{eq:ucb_nn}
    \medop
    {UCB_{\textbf{x},c} = \mathcal{S}_{\theta}(\textbf{x},c)+\alpha\sqrt{g_{\theta}(\textbf{x},c)^{T}   \textbf{D}^{-1}  g_{\theta}(\textbf{x},c)}}
\end{equation}
That is, we replace the linear model $f_{\theta}(\textbf{x}, c)$ and its derivative $f'_{\theta}(\textbf{x},c)$ by a neural network $\mathcal{S}_{\theta}(\textbf{x}, c)$ and its gradients $g_{\theta}(\textbf{x},c)$, 
{with $\theta$ and $\textbf{D}$ to be trained, as in the standard UCB.}

\fakeparagraph{Bandit Training}
\algref{alg:nnucb} illustrates how to train a contextual bandit $\mathcal{B}_{\theta, \textbf{D}}$ over time, \ie learn the parameters $\theta$ and $\textbf{D}$ defined in \equref{eq:ucb_nn} of our NN-enhanced UCB.
We first observe the working status $\textbf{x}_b$ and then preferentially explore the capacity with maximum upper confidence bound of the expected reward based on \equref{eq:ucb_nn} (lines 6 to 10).
Next, we update the covariance matrix $\textbf{D}$ using the same extension as \equref{eq:ucb_nn}, \ie replacing the derivative $f'_{\theta}(\textbf{x},c)$ of the linear model by gradients  $g_{\theta}(\textbf{x},c)$ of neural networks (line 12).
Afterwards we observe the actual workload $w_b$ as well as the reward $s_b$ (line 13).
We make such explorations and observations for a batch size of $batchSize$ (preset as 16), where we store the triples $(\textbf{x}, w, s)$ in buffer $ob$.
After collecting observations for $batchSize$ batches, we update the parameters $\theta$ by minimizing the following loss.
\begin{equation} \label{eq:loss}
        \medop
        {\mathcal{L}(\theta) = \sum_{o\in ob} \Vert  \mathcal{S}_{\theta}(\textbf{x}_o, w_o) - s_o \Vert^{2}  + \lambda \Vert \theta \Vert^2_{2}}
\end{equation}
where $(x_o, w_o, s_o)$ is an observation and $\lambda$ is the regularization parameter to avoid over-fitting.

\fakeparagraph{Estimating Capacity with Bandit}
After training the bandit $\mathcal{B}_{\theta, \textbf{D}}$ as \algref{alg:nnucb}, it can be used to estimate the workload capacity for subsequent batches. 
That is, given a working status $\textbf{x}$, we choose the workload capacity $c$ with the maximum upper confidence bound calculated by \equref{eq:ucb_nn}.
For ease of presentation, we use $\mathcal{B}.estimate(\textbf{x})$ to represent the workload capacity estimation, which chooses a suitable workload capacity given working status $\textbf{x}$.

\begin{algorithm}[t]
    \caption{\small NN-enhanced UCB}\label{alg:nnucb}
    \SetAlgoLined
    \KwIn{Regularization parameter $\lambda$, upper confidence bound coefficient $\alpha$, and batch size $batchSize$
    }
    \KwOut{Contextual bandit $\mathcal{B}_{\theta,\textbf{D}}$
    }

     Initialize deviation matrix $\textbf{D} \leftarrow \lambda \textbf{I}$\;
     Initialize observation buffer $ob \leftarrow []$\;
     Initialize parameters $\theta$ with Gauss Distribution\;
        
    \For{each trial $t \in \mathcal{T}$}{
        
        $//$ Explore the workload capacity\; 
        Observe the broker's working status $\textbf{x}$\;
        \For{each candidate capacity $c\in \mathcal{C}$}{
             $U_{\textbf{x},c} \leftarrow \mathcal{S}_{\theta}(\textbf{x},c)+\alpha\sqrt{g_{\theta}(\textbf{x},c)^{T} \textbf{D}^{-1}  g_{\theta}(\textbf{x},c)}$\;
             Choose $b$'s capacity $c^{*} \leftarrow \mathop{\arg\max}\limits_{c}  U_{\textbf{x},c}$\;
        }

        $//$ Update upper confidence and reward function\;
        $\textbf{D} \leftarrow \textbf{D} + g_{\theta}(\textbf{x},c)\cdot g_{\theta}(\textbf{x},c)^{T}$ \;
        Observe broker's workload $w$ and the reward $s$\;
        Store observed triple $(\textbf{x}, w, s)$ in buffer $ob$\;
    
        \If{$ob$.size $=$ $batchSize$}{
            Define loss function $\mathcal{L}(\theta)$ as $\sum_{o\in ob} \Vert  \mathcal{S}_{\theta}(\textbf{x}_o, c_o) - s_o \Vert^2  + \lambda \Vert \theta \Vert^2_{2}$\;
        	$\theta \leftarrow \theta -  \nabla_{\theta}\mathcal{L}(\theta)$\;
        	Clear observation buffer $ob \leftarrow []$\;
        	
        }
    }
    \KwRet contextual bandit $\mathcal{B}_{\theta,\textbf{D}}$\;
\end{algorithm}
   
\subsection{Personalized Workload Capacity Estimator}
As previously mentioned, the contextual bandit only learns a generic capacity estimator for all brokers, yet the workload capacity estimation may be broker-specific.
We enable personalized workload capacity estimation by fine-tuning the neural network $\mathcal{S}_{\theta}(\textbf{x}, c)$ in \equref{eq:ucb_nn} on broker-specific data.

Concretely, we first train a base reward mapping function $\theta^{base}$, \ie the neural network defined in \equref{eq:nn}, on the observations $\cup_{b\in B} \mathcal{T}_b$ of all brokers.
Then we copy the first $L-1$ layers of $\theta^{base}$ to the broker-specific reward mapping function $\theta_{b}$ of broker $b$.
Afterwards, we freeze the first $L-1$ layers of $\theta_{b}$, and fine-tune the last full-connected layer based on the broker's observations $\mathcal{T}_b$ following \algref{alg:nnucb}.
This way, we obtain the personalized reward mapping function.

\subsection{Effectiveness of NN-Enhanced UCB}
We now theoretically analyze the effectiveness of our NN-enhanced UCB policy for contextual bandit based workload capacity estimation.
We assess the effectiveness via the \textit{regret}, a standard performance metric for bandits \cite{bandit_book}.

The regret is defined as the difference in total rewards between the optimal policy and a learned decision-making policy.
For our workload capacity estimation, the regret can be defined as the difference between the sum of sign-up rates with ideal capacities and the sum of sign-up rates with the estimated capacities, which is formalized as follows,
\begin{equation}
    \medop
    {\mathcal{R}egret = \sum_{t=1}^{|\mathcal{T}|} \max_{c\in \mathcal{C}} s(\textbf{x}_b, c) - \sum_{t=1}^{|\mathcal{T}|} s_t}
\end{equation}
where $\mathcal{T}$ is the trial triples, $\max_{c\in \mathcal{C}} s(\textbf{x}_b, c)$ is the ideal rewards under working status $\textbf{x}_b$, and $s_t$ is rewards produced by our NN-enhanced UCB algorithm.

We have the following claim on the regret of our NN-enhanced UCB algorithm.
\begin{theorem} \label{thm:regret}
For a contextual bandit adopting NN-enhanced UCB algorithm with an $L$-layer MLP network, if there are $|\mathcal{C}|$ candidate capacities, the regret of $n$ batches is no more than $\frac{n|\mathcal{C}|\xi^{L}}{\pi^{L-1}}$, where \shuyue{$\pi \approx 3.14$ is the circular constant} and $\xi$ is the maximum single value of parameters in the MLP model.
\end{theorem}

\begin{proof}
Let $c^{*}_{t} =\mathop{\arg\max}\limits_{c} S_{\theta^{*}}({\textbf{x}_t,c})$ denote the optimal capacity under the context $\textbf{x}_t$ and $\hat{c}_{t} = \mathop{\arg\max}\limits_{c} U_{\textbf{x}_t,c}$ in batch $t$.
Firstly, we analyze the instantaneous regret $r_t$ in batch $t$, which is defined as:
\begin{equation}
    \medop
    {r_t = S_{\theta^{*}}(\textbf{x}_t,c^{*}) - S_{\theta^{*}}(\textbf{x}_t,\hat{c})}
\end{equation}
    
Using the \textit{Lipschitz continuous} \cite{NIPS18_Lip} to bound instantaneous the regret of the bandit, we have, 
\begin{equation}
    \shuyue{
    \medop 
    {\Vert S_{\theta^{*}}(\textbf{x}_t,c^{*}) - S_{\theta^{*}}(\textbf{x}_t,\hat{c}) \Vert_{2} \le Lip(S_{\theta^{*}}) \Vert c^{*} - \hat{c} \Vert_{2}}
    }
\end{equation}
where \shuyue{$\Vert \cdot \Vert_{2}$ denotes the $L_{2}$-norm} and $Lip(S_{\theta^{*}})$ is called the \textit{Lipschitz constant} of $S_{\theta^{*}}$.

Then, we analyze the upper bound of {Lipschitz constant} $Lip(S_{\theta^{*}})$ of the reward function $S_{\theta^{*}}$.
According to Lemma 2 in \cite{NIPS18_Lip}, we can construct its upper bound,
\begin{equation} \label{eq:Lip_bound}
    \shuyue{
    \medop{
    {Lip(S_{\theta^{*}}) \le \frac{ \Vert W_L \text{diag}(\sigma_{L-1})W_{L-1}\cdots \text{diag}(\sigma_{1})W_1\Vert_{2}}{\pi^{L-1}}}
    }
    }
\end{equation}
where \shuyue{$\pi \approx 3.14$ is the circular constant}.
If we take the \textit{ReLU} as the activation function $\sigma$, we have $\Vert \sigma \Vert_{2} \le \Vert \text{I} \Vert_{2}$, where $\text{I}$ is the identity matrix. 
Thus, the upper bound of Lipschitz constant $ Lip(S_{\theta^{*}})$ can be rewritten as,
\begin{equation}
    \shuyue{
    \medop{
        Lip(S_{\theta^{*}}) \le \frac{ \Pi_{i=1}^{L} \Vert W_i \Vert_{op} \cdot \Pi_{i=1}^{L-1}\Vert \text{I} \Vert_{2}}{\pi^{L-1}}
        \le \frac{\xi^{L}}{\pi^{L-1}}
    }
}
\end{equation}
\shuyue{
where $\Vert W_i \Vert_{op}$ is the operator norm of $W_i$, \ie, the largest single value of $W_i$ ($\Vert W_i \Vert_{op} \le \xi$).
Since both $c^{*}$ and $\hat{c}$ belong to $\mathcal{C}$, we have $\Vert c^{*} - \hat{c}\Vert_{2} \le |\mathcal{C}|$.
}
Thus, for instantaneous regret $r_t$, we can give its upper bound,
\shuyue{
\begin{equation}
       \medop{ r_{t} = S_{\theta^{*}}(\textbf{x}_t,c^{*}) - S_{\theta^{*}}(\textbf{x}_t,\hat{c}) 
       \le Lip(S_{\theta^{*}}) \Vert c^{*} - \hat{c} \Vert_{2}
        \le \frac{|\mathcal{C}|\xi^{L}}{\pi^{L-1}}
        }
\end{equation}
}
Finally, we prove that the regret bound for our NN-enhanced UCB policy over $n$ batches is,
\begin{equation}
    \medop
    {\mathcal{R}egret = \sum_{t=1}^{n} r_{t} \le \frac{n|\mathcal{C}|\xi^{L}}{\pi^{L-1}}}
\end{equation}
\vspace{-1em}
\end{proof}

\fakeparagraph{Discussions}
We draw two practical notes from \thmref{thm:regret}.
\begin{itemize}
    \item 
    Setting a suitable number of candidate capacities is beneficial to select an optimal decision. 
    We empirically determine the range of candidate sets based on historical observations in \secref{sec:motivation} and do not explore the workload capacity with a prominent low sign-up rate.
    \item 
    Although a deeper network may model more complex relationships between a broker's performance and working status, it may also prevent the bandit from choosing the optimal workload capacity.
    We empirically adopt a 3-layer MLP network in our NN-enhanced UCB algorithm to balance the complexity of the neural network and the effectiveness of workload estimation.
    \end{itemize}

\section{Capacity-Based Assignment} \label{sec:method_assignment}
Now we present the assignment module of LACB, which takes the estimated capacity as input and makes assignments by accounting for both the capacity constraint and the dependency of assignments across {batches}.

\subsection{Batched Assignment as Markov Decision Process}
Unlike previous studies \cite{ ICDE19_online_task, ICDE19_batch} that \textit{independently} make assignments for each batch, we propose to match brokers in a more holistic view by modeling the assignments over time as a \textit{Markov Decision Process (MDP)} \cite{bandit_book}.
Such modeling accounts for the dependencies of assignments across batches (\ie residual capacities of brokers over time) and potentially results in a higher total utility.

A standard MDP model consists of four elements: the \textit{state}, \textit{action}, \textit{state transition}, and \textit{reward}.
We explain these elements in the context of broker assignment below.
\begin{itemize}
    \item \textbf{State}. 
    As our assignment decision is capacity-aware, we define the state of each broker as $cr_b \in [0, c_b]$, where $cr_b$ and $c_b$ denotes the broker's residue capacity and the workload capacity, respectively.
    \item \textbf{Action}. 
    In batch $i$, the action is an assignment policy $\mathcal{M}_{r,b}^{(i)}$ for each request $r$ and each broker $b$. 
    If $\mathcal{M}_{r,b}^{(i)}=1$, broker $b$ is assigned to request $r$, and $\mathcal{M}_{r,b}^{(i)}=0$ otherwise.
    In this work, we adopt a value function guided assignment policy (see \secref{sec:method_VFGA}).
    \item \textbf{State Transition}.
    The state of a broker changes as the result of action $\mathcal{M}_{r,b}^{(i)}$. 
    If $\sum \mathcal{M}_{r,b}^{(i)} = 0$, the state of broker $b$ remains the same.
    Otherwise, state $cr_b$ will transit to $cr_b - \sum \mathcal{M}_{r,b}^{(i)}$.
    That is, the residue capacity of the broker is reduced by the number of requests assigned to him/her. 
    \item \textbf{Reward}.
    The rewards of an action is defines as the utility of all brokers in batch $i$ (a.k.a batch utility), \ie $r(\mathcal{M}^{(i)}) = \sum_{r,b}u_{r,b}M^{(i)}_{r,b}$.
    Note that the reward of the MDP model differs from the reward of a bandit (in \secref{sec:method_estimate}), and the latter is the accumulative utility of a single broker, \ie $s_b = \sum_i \sum_{r} u_{r,b}\mathcal{M}^{(i)}_{r,b}$.
\end{itemize}
Note that we formulate the batched assignment as an MDP model (and adopt reinforcement learning based solution) rather than a bandit algorithm because the latter is unfit for long-term planning with state transitions \cite{bandit_book}.

\subsection{Capacity-Aware Assignment} 
\label{sec:method_VFGA}
Given the MDP model above, we utilize a capacity-aware value function to guide the assignments.

\fakeparagraph{Capacity-Aware Value Function}
It is common to make decisions in an MDP by learning the value function of a state \cite{ICDE19_batch, KDD18_batch}.
In this work, we define a capacity-aware value function $\mathcal{V}(i,cr)$, which represents the expected utility of the broker after batch $i$, where $cr$ is the broker's residue capacity.
Such a capacity-aware value function captures the long-term utility of brokers with different residue capacities.
We then use Q-learning \cite{ICDE19_batch}, a classical method with relatively low time overhead, to train the capacity-aware value function, which is based on the \underline{T}emporal-\underline{D}ifference (TD) equation below.
\begin{equation} \label{eq:value_function}
    \medop
    {\mathcal{V}(cr) \leftarrow
             \mathcal{V}(cr) + \beta[u+\gamma \mathcal{V}(cr')- \mathcal{V}(cr)]}
\end{equation}
where $cr/cr'$ are current/transited state after taking an action, $u$ is the reward of an action, $\beta$ and $\gamma$ are the learning rate and discount factor of the TD, respectively.
\begin{algorithm}[bth]
    \caption{Value Function Guided Assignment}
    \label{alg:assignment}
    \SetAlgoLined
    \SetAlgoVlined
    \KwIn{Brokers $B=\{b_1,b_2, ..., b_{|B|}\}$, requests $R =\{R^{(1)}, R^{(2)} ,..., R^{(|I|)}\}$ over $I$ intervals, brokers' exclusive bandits $\{\mathcal{B}_1$, $\mathcal{B}_{2}$,..., $\mathcal{B}_{|B|}\}$ } 
           
    \KwOut{Matching results $\mathcal{M}$}
    
    \For{each broker $b$ in $B$}{
            $c_b \leftarrow \mathcal{B}_{b}.estimate(\text{x}_b)$\;
    }
        
    \For{each interval $i\in I$}{
        Get available brokers $B_{+}=\{b| w_b<c_b \wedge b\in B\}$ \;
        \For{each candidate pair $(r,b)\in R^{(i)}\times B_{+}$}{
         \begin{equation}
            u'_{r,b}\leftarrow \left\{
            \begin{aligned}
                &u_{r,b}+0, &\text{if } f_b \le \delta  \\
                &u_{r,b}+\gamma\mathcal{V}(cr')-\mathcal{V}(cr), & \text{if} f_b > \delta
            \end{aligned} \right.
            \notag
         \end{equation}    
        }
        Execute the Kuhn-Munkres algorithm based on the refined utility $\mathcal{M}^{(i)}=KM(u',R^{(i)},B_{+})$\;
        \For{each broker $b$ in $B_{+}$}{
        $w_b \leftarrow w_b + \sum_{r} \mathcal{M}^{(i)}_{r,b}$\;
        Update capacity-aware value function based on the  Temporal-Difference equation \equref{eq:value_function}\;
        }
    }
        
    \For{each broker $b$ in $B$}{
        $s_b \leftarrow \sum_{i} \sum_{r} u_{r,b}\mathcal{M}^{(i)}_{r,b}$\;
        $\mathcal{B}_{b}.update(\text{x}_b, w_b, s_b)$\;
    }
		 
    \KwRet {$\mathcal{M}=\cup_{i\in I} \mathcal{M}^{(i)}$}

\end{algorithm}

\fakeparagraph{Value Function Guided Assignment (VFGA)}
We can now leverage the above capacity-aware value function to assign brokers from a global view to maximize the total utility.
\algref{alg:assignment} shows the overall assignment algorithm.
First, we determine the personalized capacity of brokers from the contextual bandit $\mathcal{B}_b$ (see \secref{sec:method_estimate}).
In lines 4-14, we make assignment in each batch.
Specifically, in line 5, we first select a set of available brokers $B_{+}$, whose workload $w_b$ is lower than his/her capacity $c_b$.
Then we update the utility of each candidate matching pair as follows.
\begin{equation} \label{eq:refine}
        \medop
        {u'_{r,b} =} 
        \left\{
        \begin{aligned}
            &
            \medop
            {u_{r,b}+0,} &\medop{\text{if } f_b \le \delta}  \\
            &
            \medop
            {u_{r,b}+\gamma\mathcal{V}(cr')-\mathcal{V}(cr),}  &
            \medop
            {\text{if } f_b > \delta}
        \end{aligned}
        \right.
\end{equation}
where $u_{r,b}/u'_{r,b}$ is the original/refined utility respectively, $cr = c_b - w_b$ and $cr' = cr - 1$.
Note that only top brokers may reach their workload capacity, and we only use the value function for top brokers whose frequency $f_b$ of reaching capacity is more than $\delta$, where $\delta$ is a positive number close to 1.
In line 9, we run the classical Kuhn–Munkre (KM) \cite{KM_alg} algorithm on a bipartite graph with the refined utility $u'_{r,b}$ and return the assignment policy $\mathcal{M}^{(i)}$.
In lines 10-12, the workloads of the assigned brokers and the capacity-aware value function are updated according to \equref{eq:value_function}.
In lines 15-17, once we have assigned requests in all batches, we collect the brokers' workloads and performance as feedback to update the bandit of each broker.
As aforementioned, we update the parameters of bandits whenever the observation buffer is full.
Otherwise, we only add new observations to the buffer.
Finally, VFGA returns the assignment results $\mathcal{M}$.

\shuyue{\fakeparagraph{Discussions}
    We make two notes on the assignment module.
}
\begin{itemize}
    \item Since the number of requests $|R|$ is usually smaller than that of brokers $|B|$, a common practice is to add some dummy vertices to the smaller part to construct a balanced bipartite graph \cite{VLDB21_Match, KBS19_SkewMatch, Skewed_Match12}.
    By adding $|B|-|R|$ dummy vertices, we obtain a balanced one with $|B|$ vertices on both sides and can execute the classical KM algorithm.
    \shuyue{\item Once a client is unsatisfied with the assigned broker, she/he can appeal to the platform for another broker.
    The platform sets the utility between the client and the assigned broker to 0, restores the broker's workload, and chooses another broker in the next time interval.
    }
\end{itemize}

\begin{figure}[htb]
    \centering
    \vspace{-1em}
    \includegraphics[width=\linewidth]{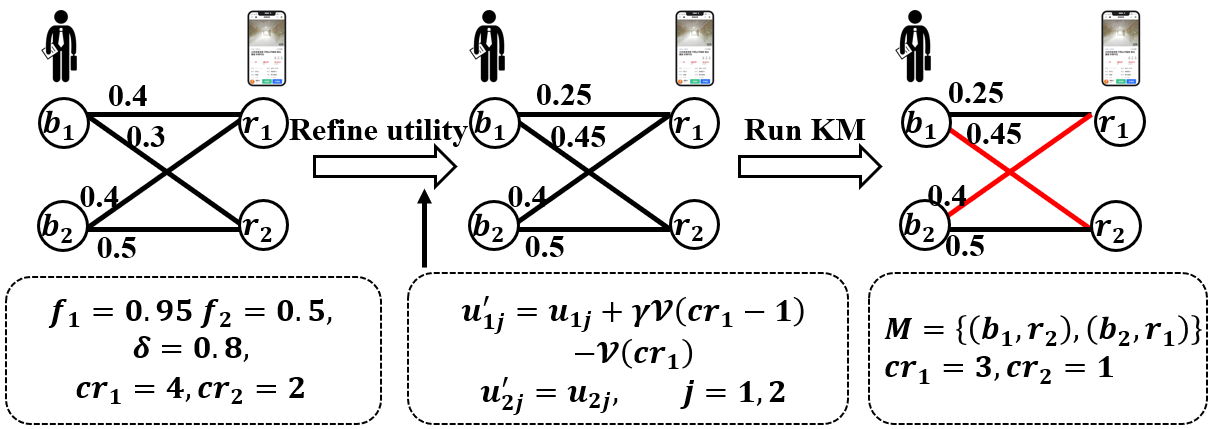}
    \caption{\small Example of the VFGA algorithm.} 
    \label{fig:VFGA}
    \vspace{-1em}
\end{figure}

\figref{fig:VFGA} shows an example of the core steps in the VFGA algorithm.
There are two brokers $b_1, b_2$ and two requests $r_1,r_2$. 
The utility of $(b_1,r_1)$, $(b_1,r_2)$, $(b_2,r_1)$ and $(b_2,r_2)$ are $0.4,0.3,0.4$ and $0.5$ respectively.
Based on \equref{eq:refine}, $b_1$ reaches the threshold $\delta=0.8$ and we refine $b_1$'s utility as $u_{11}=0.25$ and $u_{12}=0.45$, while $b_2$'s utility remains the same.
Then we run the KM algorithm to get the maximum assignment over refined weights, \ie $\{(b_1,r_2),(b_2,r_1)\}$.
Finally, residue capacities of $b_1$ and $b_2$ are updated to $3$ and $1$, respectively.

\fakeparagraph{Complexity Analysis}
For each batch, the time complexity of the VFGA algorithm is determined by the KM algorithm \cite{KM_alg}.
Since the KM algorithm is run on a balanced bipartite graph with $|B|$ vertices on both sides, the time complexity of the VFGA algorithm is $\mathcal{O}(|B|^{3})$.

\subsection{Accelerating Assignment via Broker Selection} 
\label{sec:method_assignment_optimization}
\shuyue{The real estate platform needs to assign brokers to clients as efficiently as possible, which is an essential factor in user experience.
It usually requires responding in 2 seconds for reasonable user experience \cite{ResTime}.}
There is an opportunity to accelerate the VFGA algorithm because the numbers of brokers and requests are highly imbalanced in online real estate platforms.
In each batch, the platform typically assigns thousands of brokers to only tens of requests.
If we prune the brokers that are unlikely to be matched, we can notably reduce the number of dummy vertices added to the request side, and thus lower the time complexity.

\begin{algorithm}[t]
    \caption{Candidate Broker Selection (CBS)}
    \label{alg:sort}
    \SetAlgoLined
    \KwIn{The candidate size $k$, request $r$ and brokers $B$} 
           
    \KwOut{The candidate brokers Top$^{r}_{k}$}
    \If{$|B| \le k$}{
        return $B$\;
    }
    Choose an random value $p$ from $u_{r,b}$ ($b\in B$)\;
    $LC \leftarrow \{b\in B| u_{r,b} \ge p\}$\;
    $RC \leftarrow \{b\in B| u_{r,b} < p\}$\;
    \If{$|LC| \ge k$}{
        {Top}$^{r}_{k} \leftarrow$ CBS$(k,r,LC)$
    } \Else
    {
    {Top}$^{r}_{k}\leftarrow$ LC $\cup$ CBS$(k-|LC|, r,RC)$\;
    }
    \KwRet \text{Top}$^{r}_{k}$
\end{algorithm}

\fakeparagraph{Theoretical Evidence}
We claim that only a small set of brokers are necessary for effective assignments, as shown by Theorem~\ref{thm:unbalanced} and Corollary~\ref{coro:no_loss}.
\begin{theorem}\label{thm:unbalanced}
    Given a bipartite graph $G=<U,V, E>$ ($|U|\le |V|$), and let $opt_{\mathcal{M}}(u)$ be the matched vertex of $u$ in the optimal assignment $\mathcal{M}$.
    There exists an optimal assignment $\mathcal{M}$, such that for any vertex $u \in U$, we have $opt_{\mathcal{M}}(u)\in$ Top$^{u}_{|U|}$, where Top$^{u}_{|U|}$ is a set of vertices with $|U|$ largest edge weight among all vertices connected $u$.
\end{theorem}
\begin{proof}
    The proof is by contradiction.
    We start by assuming the opposite. 
    Let $u$ be any vertex of $U$ and the vertex matched to $u$ be $v^{*}$.
    Assume there is no optimal assignment $\mathcal{M}$, such that $v^{*}=opt_{\mathcal{M}}(u) \in$ Top$^{u}_{|U|}$.
    Consider a corner case, where any other vertex of $U$ is matched to a vertex of Top$^{u}_{|U|}$ in the optimal assignment $\mathcal{M}$.
    There is still at least one unmatched vertex $v'\in$ Top$^{u}_{|U|}$.
    By the definition of Top$^{u}_{|U|}$, we have $w(u,v')\ge w(u,v^{*})$.
    If we replace $(u, v')$ with $(u, v^{*})$ and keep the other matching of $\mathcal{M}$ the same, we can construct another optimal assignment $\mathcal{M}'$, which contradicts the assumption that ``\textit{there is no optimal assignment}."
\end{proof}
\begin{corollary} 
    \label{coro:no_loss}
    Given an imbalanced bipartite graph $G=<U,V,E>$, we need at most $|U|$ candidate vertices for each vertex $u\in U$ to find an optimal assignment, \ie taking the Top$^{u}_{|U|}$ as the candidate set for any $u\in U$.
\end{corollary}

\fakeparagraph{Candidate Broker Selection Algorithm}
To efficiently select the candidate set Top$^{r}_{|R|}$ for each request $r$, we devise the Candidate Broker Selection (CBS) algorithm (see \algref{alg:sort}), which is inspired by solutions to the classical \textit{selection problem} \cite{KM_alg}.
We mainly introduce how to integrate this optimization into our VFGA.
In each batch, we execute the CBS algorithm to select the necessary brokers $\cup_{r\in R^{(i)}}$Top$^{r}_{|R^{(i)}|}$ after getting the available brokers $B_{+}$ (line 5 in \algref{alg:assignment}).
Then we can perform assignment on a much smaller bipartite graph.

\fakeparagraph{Complexity Analysis}
The expected time complexity of candidate broker selection is $\mathcal{O}(|R||B|)$. 
Based on the above optimization, we can assign brokers on the bipartite graph with $|R|$ vertices and $|R|^{2}$ edges, so the time complexity of executing the KM algorithm is reduced from $\mathcal{O}(|B|^{3})$ to $\mathcal{O}(|R|^{3})$, where $|R|$ are usually much smaller than $|B|$. 
As a result, the overall time complexity of the VFGA algorithm with CBS is $\mathcal{O}(|R|^{3}+|R||B|)$.

\begin{table}[thb]
    \centering
    \caption{ Synthetic datasets.}
    \label{tlb:syndata}
    \begin{tabular}{|c|c|}
    \hline
        \textbf{Factor} & \textbf{Setting}                                  \\ \hline
        The number of brokers $|B|$        & 500, 1000, \textbf{2000}, 5000, 10000    \\ \hline
         The number of requests $|R|$     & 10K, 20K,  \textbf{50K}, 100K, 200K    \\ \hline
         The number of covering days $Day$        & 7,  10, \textbf{14}, 17, 21 \\ \hline
        The degree of imbalance $\sigma$    &0.005, 0.01, \textbf{0.015}, 0.02, 0.05 \\ \hline 
    \end{tabular}
    \vspace{-2em}
\end{table}
\begin{table}[thb]
    \caption{Real-world datasets.}\label{tlb:city_data}
    \centering
    \begin{tabular}{cccc}
    \hline
    City   & Dates                         & Brokers & Requests \\ \hline
    City A & Aug. 1 $\sim$ Aug. 21, 2021 & 5515    & 103106     \\
    City B & Jul. 1 $\sim$ Jul. 21, 2021 & 8155    & 387339     \\
    City C & Jun. 8 $\sim$ Jun. 28, 2021 & 3689    & 74831
    \\ \hline
    \end{tabular}
    \vspace{-1em}
\end{table}

\section{Experimental Study} \label{sec:experiment}
This section presents the evaluations of LACB.
\subsection{Experiment Setup}

\fakeparagraph{Datasets}
We test our capacity-aware assignment algorithm over both synthetic and real datasets.
\begin{itemize}
    \item \textit{Synthetic Datasets}.
    We generate 2000 brokers and 50000 requests in total.
    Then we vary the number of brokers, requests, 
    covering days, and the degree of imbalance.
    The degree of imbalance $\sigma={|R|}/{|B|}$ is the ratio between requests and brokers in each batch.
    To vary $\sigma$, we keep the number of brokers the same and change the number of requests.
    \tabref{tlb:syndata} summarizes the configuration of synthetic datasets.
    We mark the default settings in bold.
    \item \textit{Real-world Datasets}.
    We collect data in three cities covering 21 days from Beike, the largest Chinese online real estate platform.
    \tabref{tlb:city_data} summarizes the statistics of real datasets.
\end{itemize}

\fakeparagraph{Implementation} 
\shuyue{
   Our experiments are conducted on a simulator of Beike, which takes the same utility function deployed and outputs the utility between requests and brokers, so that we can take such utility as the input and assess algorithms over real world matching instances.
}
\shuyue{
In the NN-enhanced UCB, we adopted a $3$-layer MLP network and set the size of input layer, hidden layer and output layers as 128, 64 and 16, respectively.
We take the \textit{ReLU} as the activation function.}
In \algref{alg:nnucb}, we set $\alpha$ to 0.001, the batch size ${batch}{Size}$ to 16 and the regularization parameter $\lambda$ to 0.001.
\shuyue{
In \algref{alg:assignment}, we set the learning rate $\beta$ to $0.25$, discount factor $\gamma$ to $0.9$ and the threshold $\delta$ to $0.8$.
}
All the experiments were conducted on Intel(R) Xeon(R) Gold 6230R CPU @ 2.10GHz with 32GB main memory.
The algorithms are implemented in Python 3. 

 \begin{figure*}[t] 
    \centering
    \subcaptionbox{Utility of varying $|B|$}{
    \centering
    \includegraphics[width=0.2\linewidth]{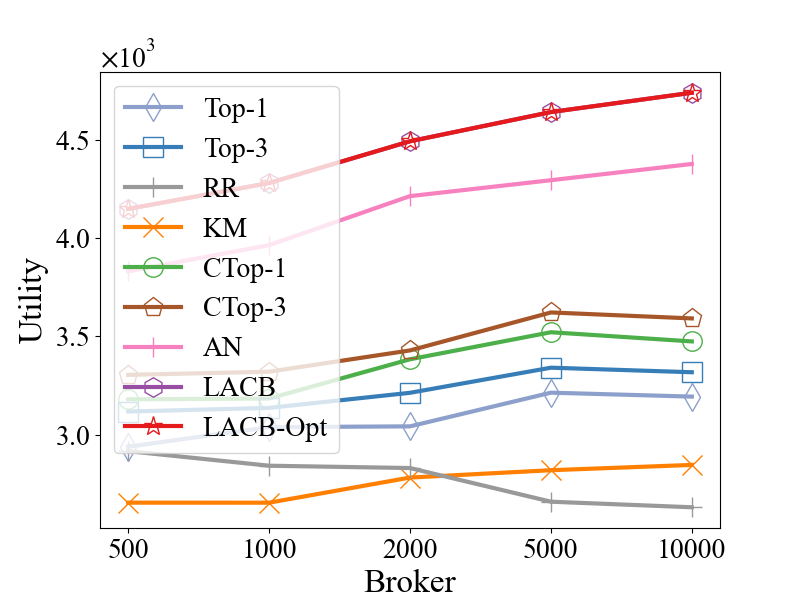}
    }
    \subcaptionbox{Utility of varying $|R|$}{
    \centering
    \includegraphics[width=0.2\linewidth]{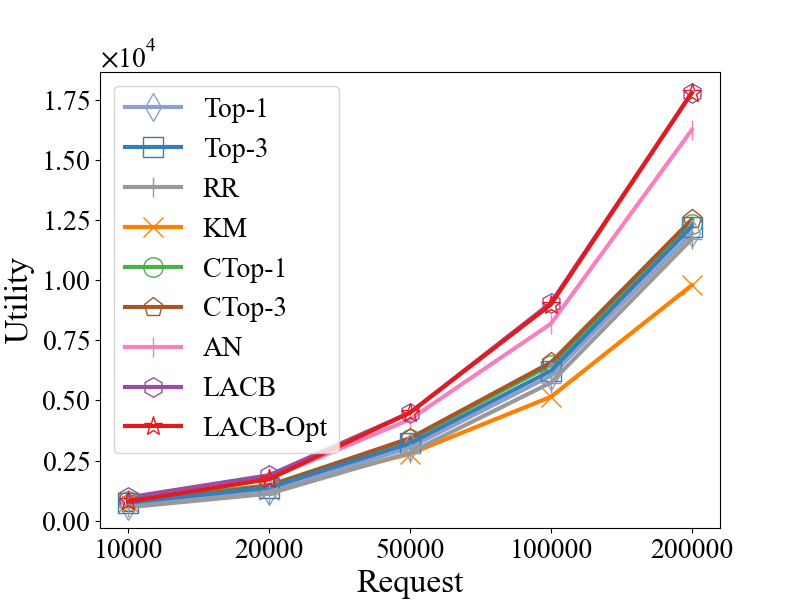}
    }
    \subcaptionbox{Utility of varying $Day$}{
    \centering
    \includegraphics[width=0.2\linewidth]{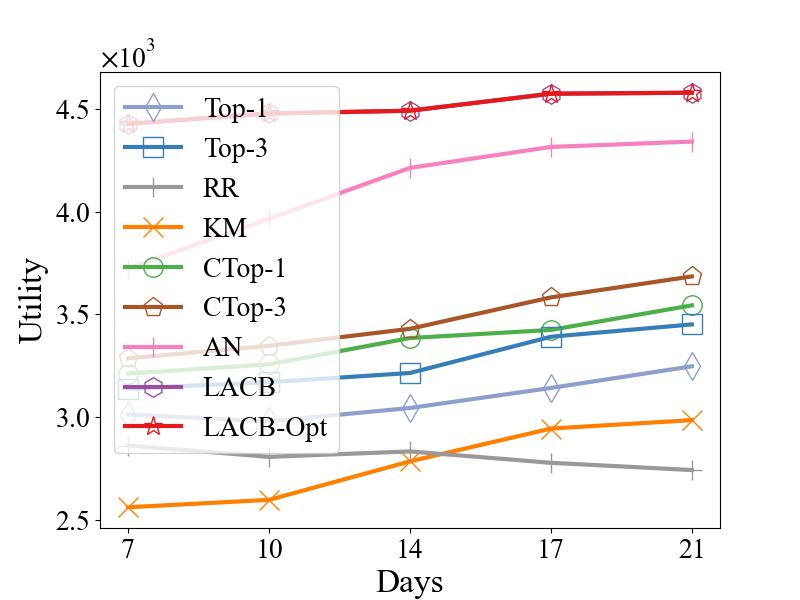}
    }
    \centering
    \subcaptionbox{Utility of varying $\sigma$}{
    \centering
    \includegraphics[width=0.2\linewidth]{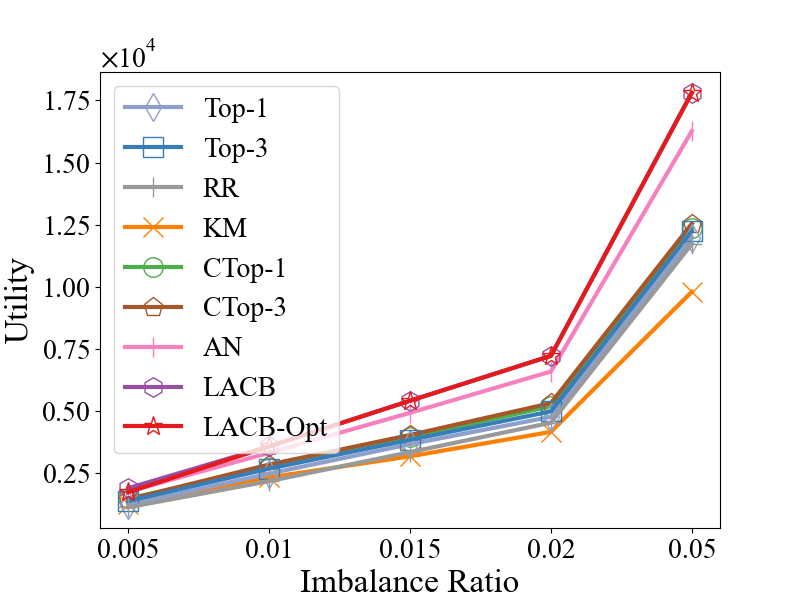}
    }
    
    \centering
    \subcaptionbox{Time of varying $|B|$}{
    \centering
    \includegraphics[width=0.2\linewidth]{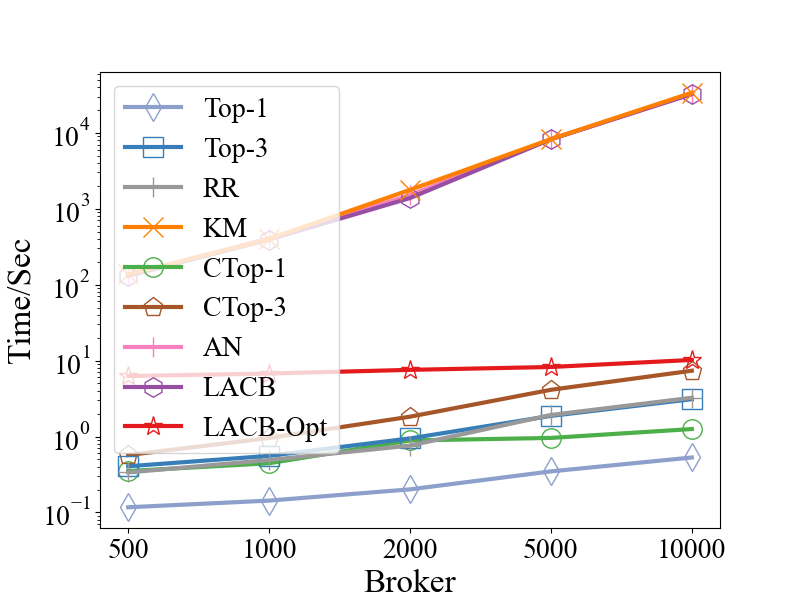}
    }
    \subcaptionbox{Time of varying $|R|$}{
    \centering
    \includegraphics[width=0.2\linewidth]{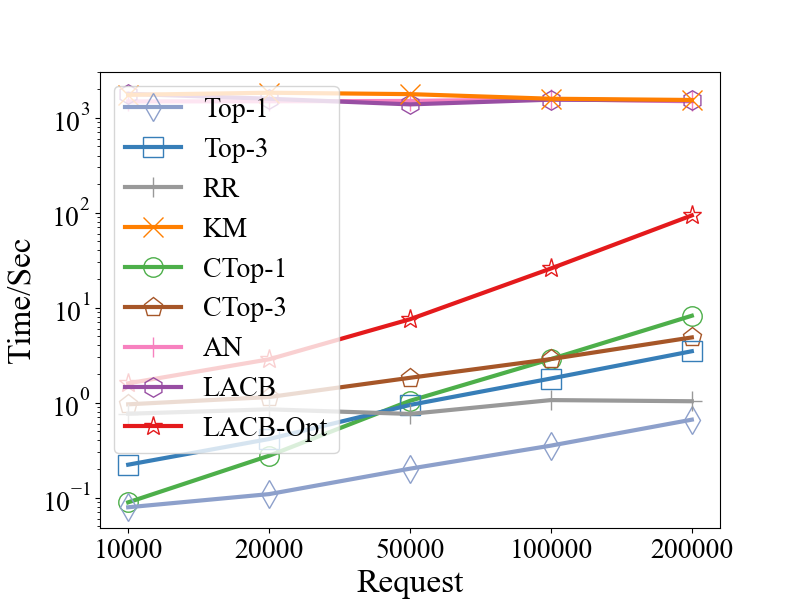}
    }
    \subcaptionbox{Time of varying $Day$}{
    \centering
    \includegraphics[width=0.2\linewidth]{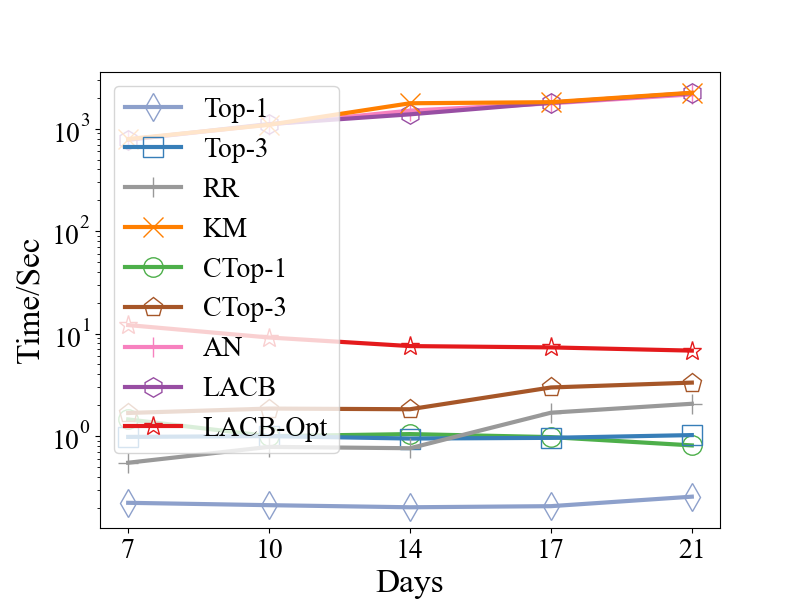}
    }
    \centering
    \subcaptionbox{Time of varying $\sigma$}{
    \centering
    \includegraphics[width=0.2\linewidth]{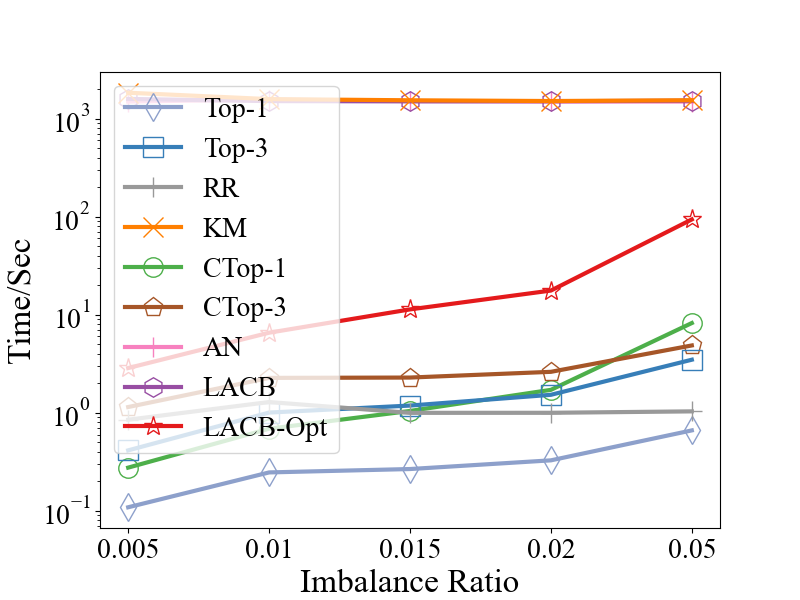}
    }
    
    \caption{\small Results on synthetic datasets.}
    \label{fig:syn_results}
    \vspace{-2em}
    \end{figure*} 
    
\fakeparagraph{Compared Algorithms} 
We compare our LACB \ie \secref{sec:method_VFGA} and LACB with CBS \ie \secref{sec:method_assignment_optimization} (denoted as LACB-Opt) with two categories of baselines.
The first category (Top-K, RR and KM) does not set explicit capacity for brokers, while the second category (CTop-K and AN) first chooses brokers' capacities and then assigns them to requests.
\begin{itemize}
    \item Top-K Recommendation (Top-K) \cite{Top-K}: 
    It ranks and returns K brokers with the highest utility.
    We evaluate both Top-1 and Top-3 recommendation.

    \item Randomized Recommendation (RR): 
    It takes the broker's service quality as the sampling weight and recommends a random broker.
    RR extends prior \textit{fair matching} algorithms \cite{TOSC18_fair} and views service quality as the fairness index.
    It can avoid the overloaded phenomenon by apportioning online requests to all brokers.
    \item Kuhn–Munkre (KM) algorithm \cite{KM_alg}: 
    It runs the KM algorithm to assign brokers to requests in each batch.
    \item Constrained Top-K (CTop-K) \cite{CIKM17_capacity_rec}: 
    It is an extension of Top-K where CTop-K observes the relations of workload and the sign-up rate at city levels (\figref{fig:overall_agents}) and empirically chooses a capacity for all brokers.
    \shuyue{The empirical workload capacity is set as 45, 55 and 40 for brokers in City A, City B and City C, respectively.}
    We test CTop-K with both K=1 and K=3.

    \item Assignment with NeuralUCB (AN): 
    It explores the workload capacity by NeuralUCB \cite{ICML20_nn_bandit} and assigns brokers to requests by the KM algorithm \cite{KM_alg}.
    
    \end{itemize}

\subsection{Performance on Synthetic Datasets}
In this set of experiments, we test the impact of different parameters on the performance of different algorithms.

\fakeparagraph{Impact of \# brokers} 
The first column of \figref{fig:syn_results} shows the results of varying $|B|$.  
For the total utility, our LACB and LACB-Opt dominate other baselines, including Top-K, CTop-K, KM, RR and AN.
We also observe that the utility of Top-K Recommendation decreases as $|B|$ increases, indicating that providing more brokers will not improve the total utility due to the overloaded phenomenon.
Finally, LACB-Opt achieves the same utility as LACB, which is consistent with our theoretical analysis in \corref{coro:no_loss}, \ie the candidate broker selection does not sacrifice the total utility.
In terms of the running time, as $|B|$ increases, KM, AN, and LACB become inefficient due to their cubic time complexity, yet the running time of Top-k, RR and CTop-k only increases marginally.  
The running time of LABT-Opt remains stable since its time complexity is mainly decided by the number of requests and is faster than other KM-based algorithms (KM, AN, LACB).

 \begin{figure*}[t]
        \centering
        \subcaptionbox{City A}{
        \centering
        \includegraphics[width=0.225\linewidth]{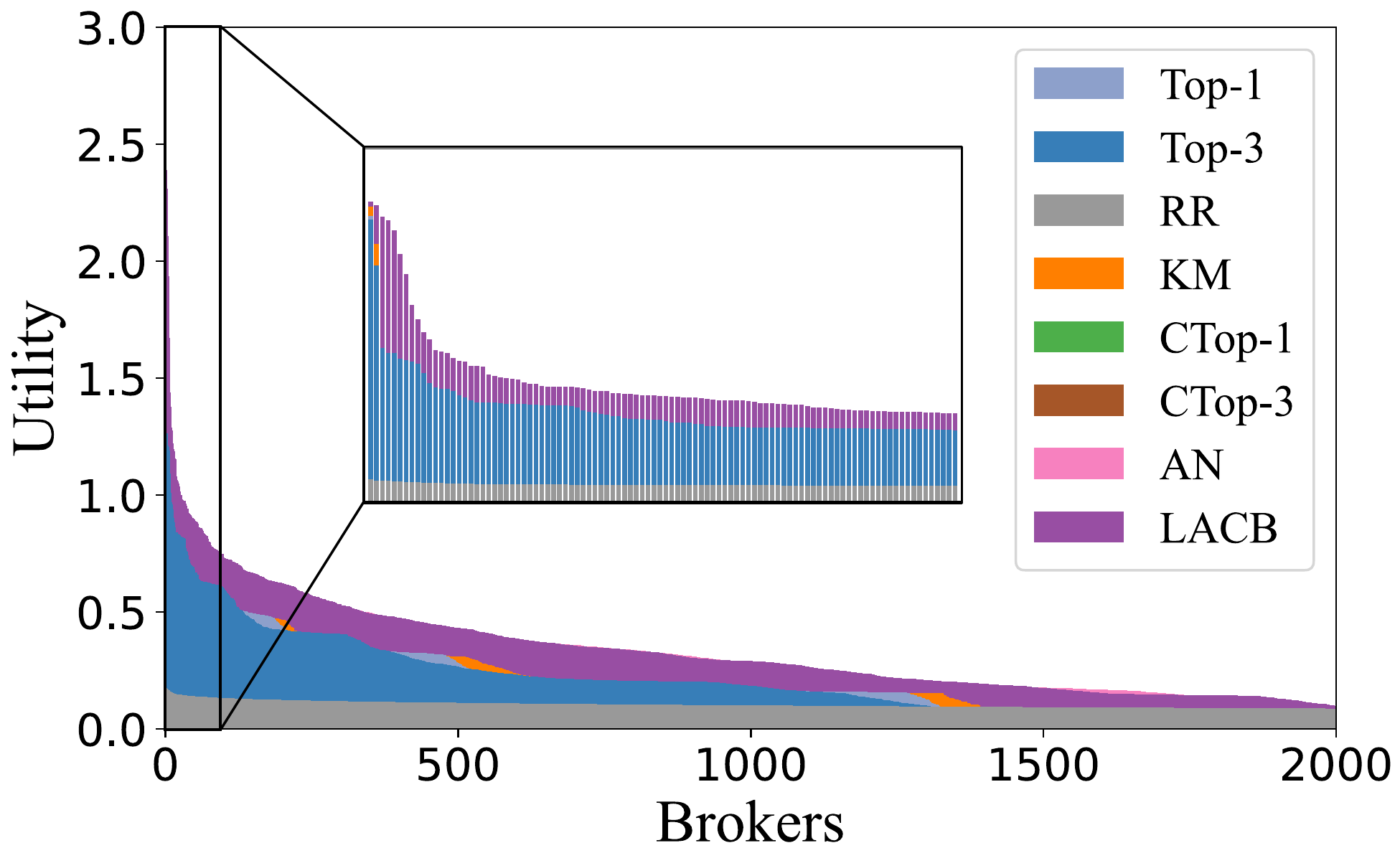}
        }\hfill
        \subcaptionbox{City B}{
        \centering
        \includegraphics[width=0.225\linewidth]{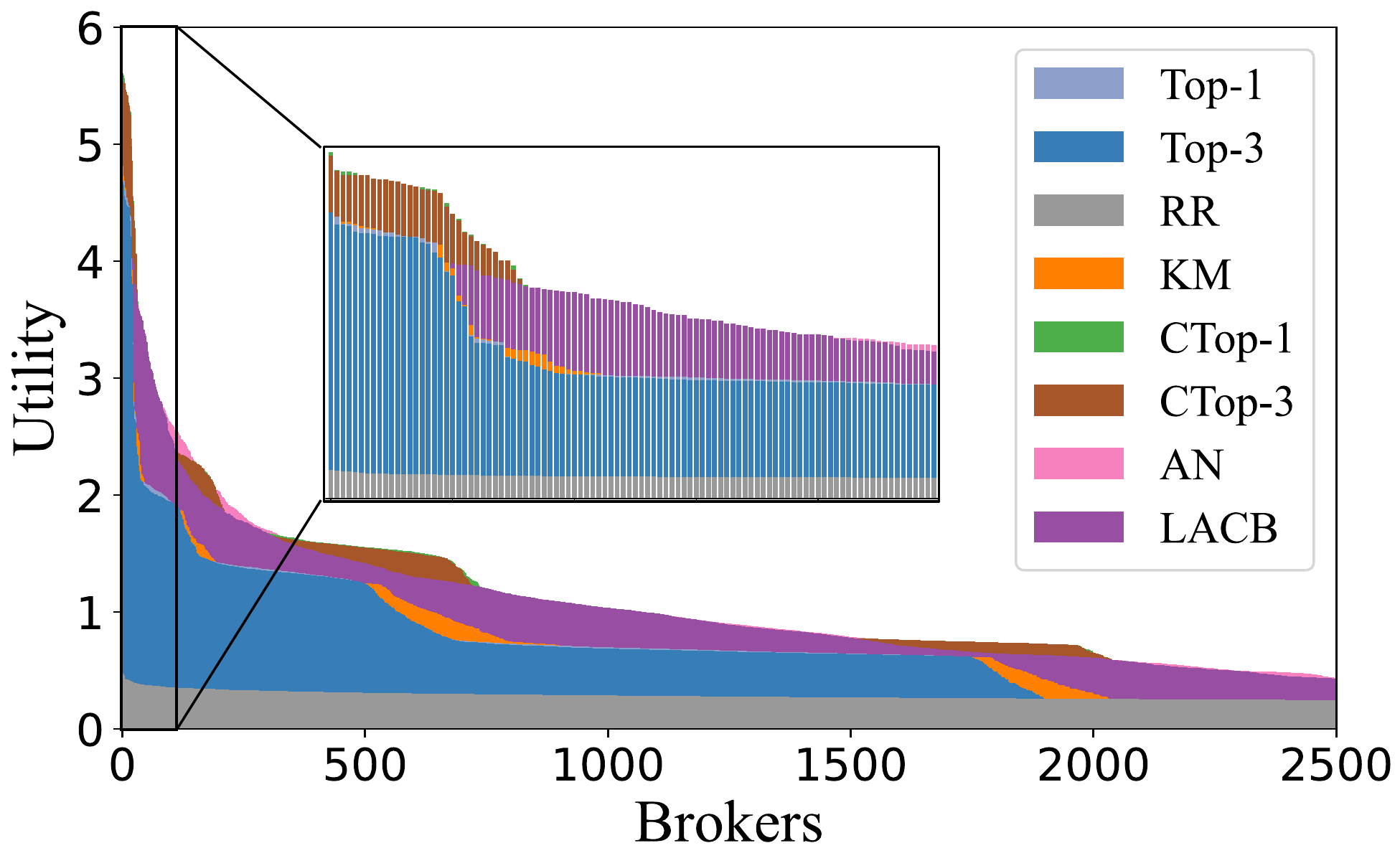}
        }\hfill
        \subcaptionbox{City C}{
        \centering
        \includegraphics[width=0.225\linewidth]{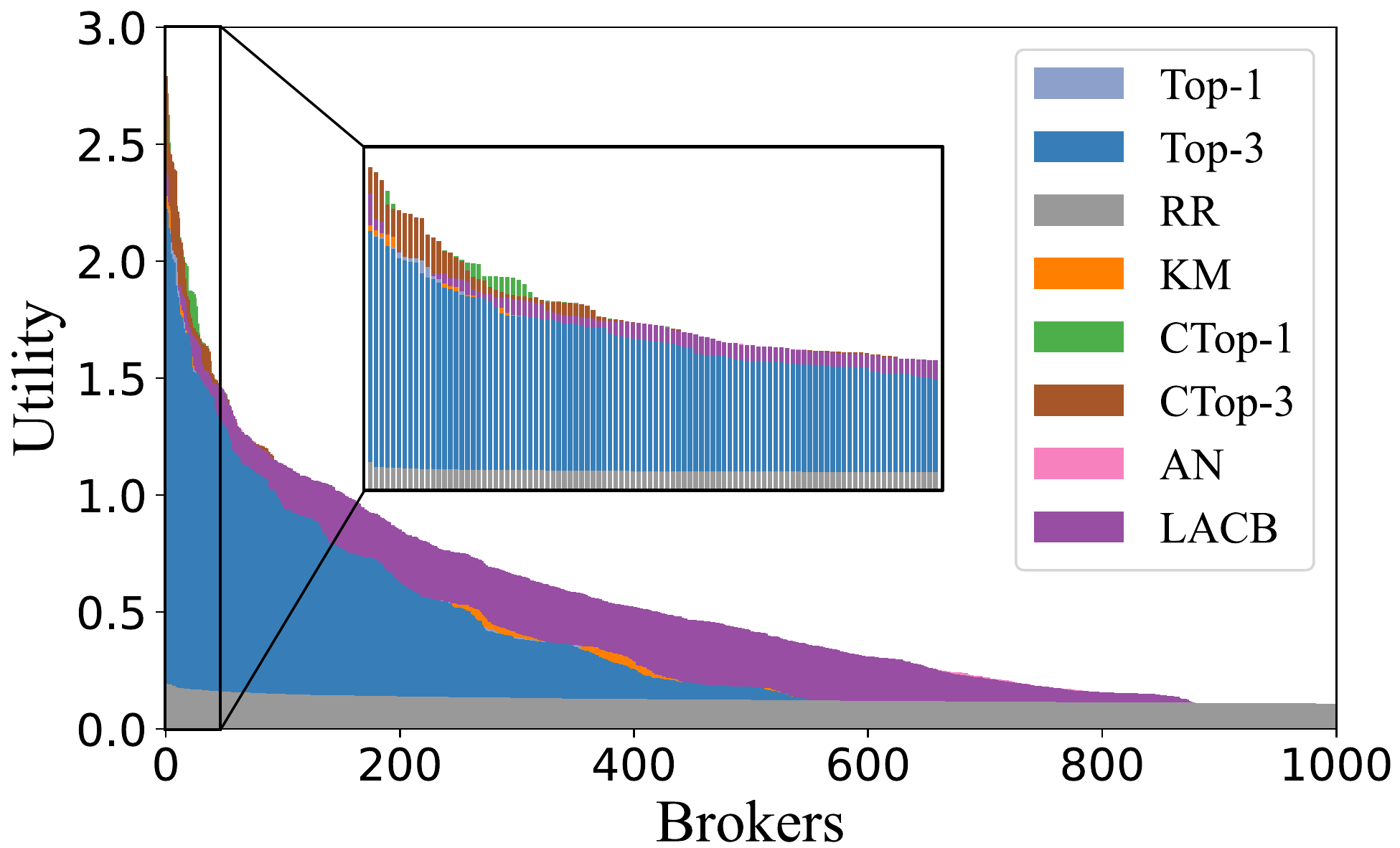}
        }
        \centering
        \caption{\small The utility distribution of all compared algorithms.
        We have a further look on top brokers' utility. }
        \label{fig:exp_utility_dis}
        \vspace{-0.5em}
        \end{figure*} 
        \vspace{-0.5em}
        \begin{figure*}[thb]
        \centering
        \subcaptionbox{City A}{
        \centering
        \includegraphics[width=0.225\linewidth]{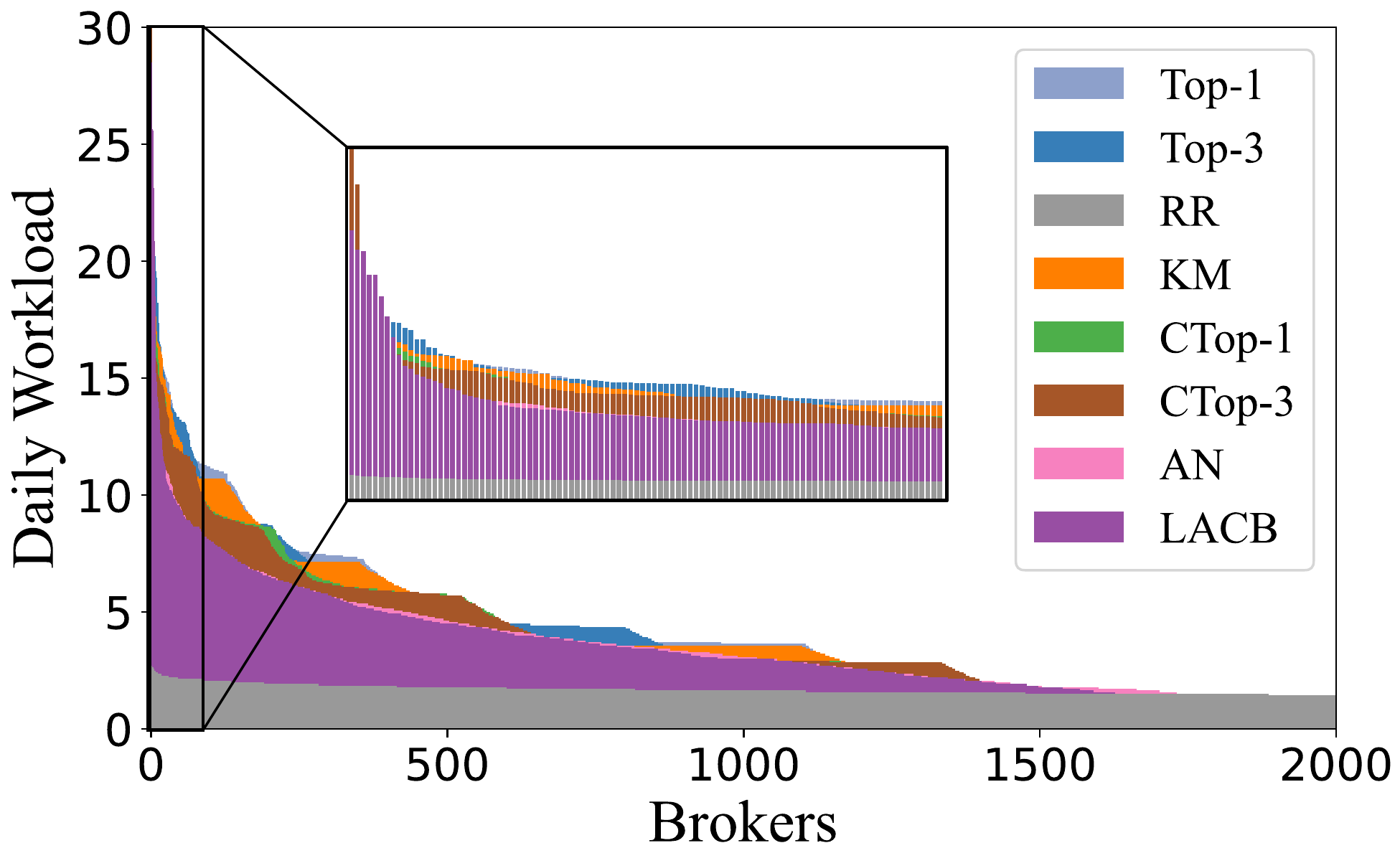}
        }\hfill
        \subcaptionbox{City B}{
        \centering
        \includegraphics[width=0.225\linewidth]{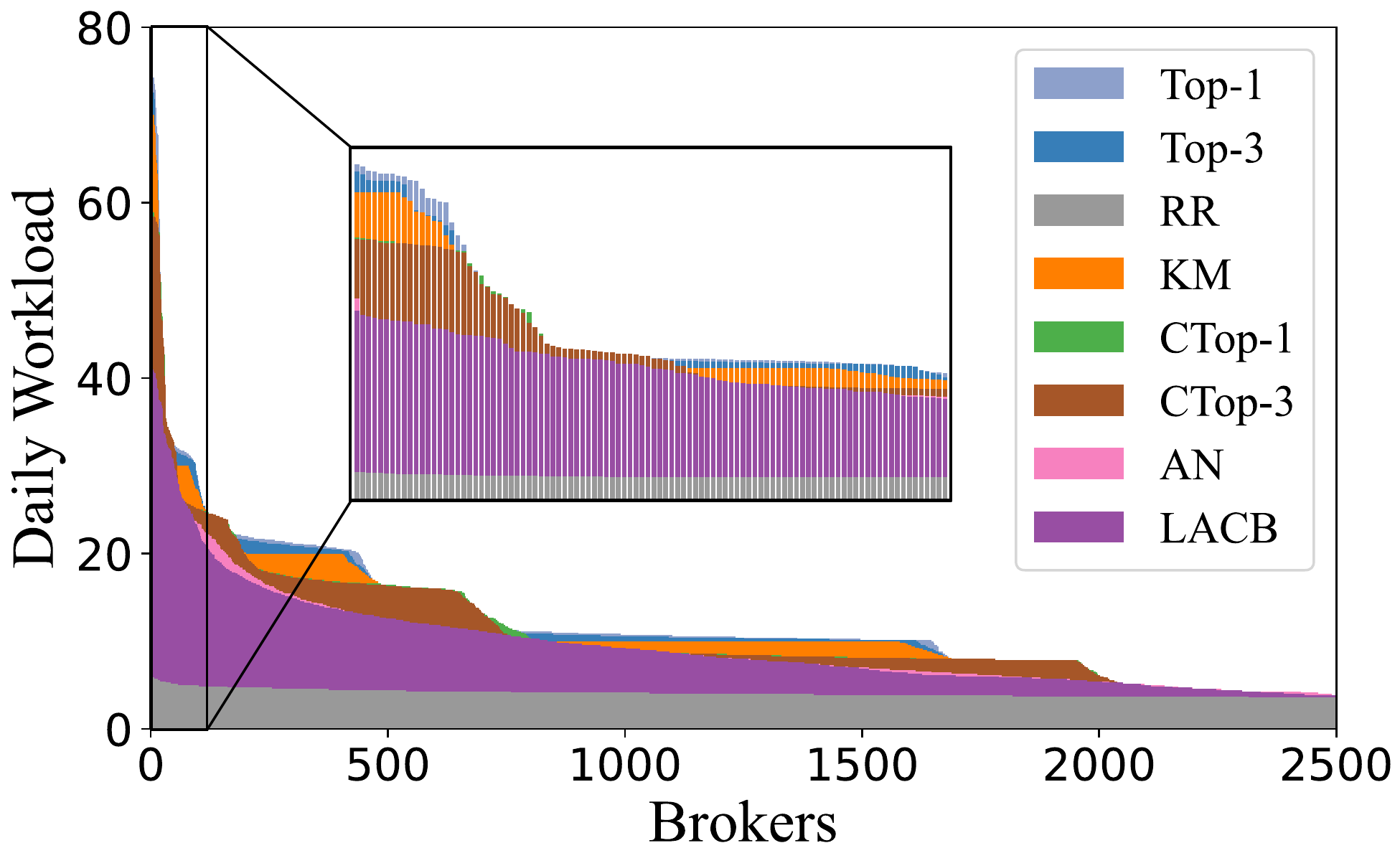}
        }\hfill
        \subcaptionbox{City C}{
        \centering
        \includegraphics[width=0.225\linewidth]{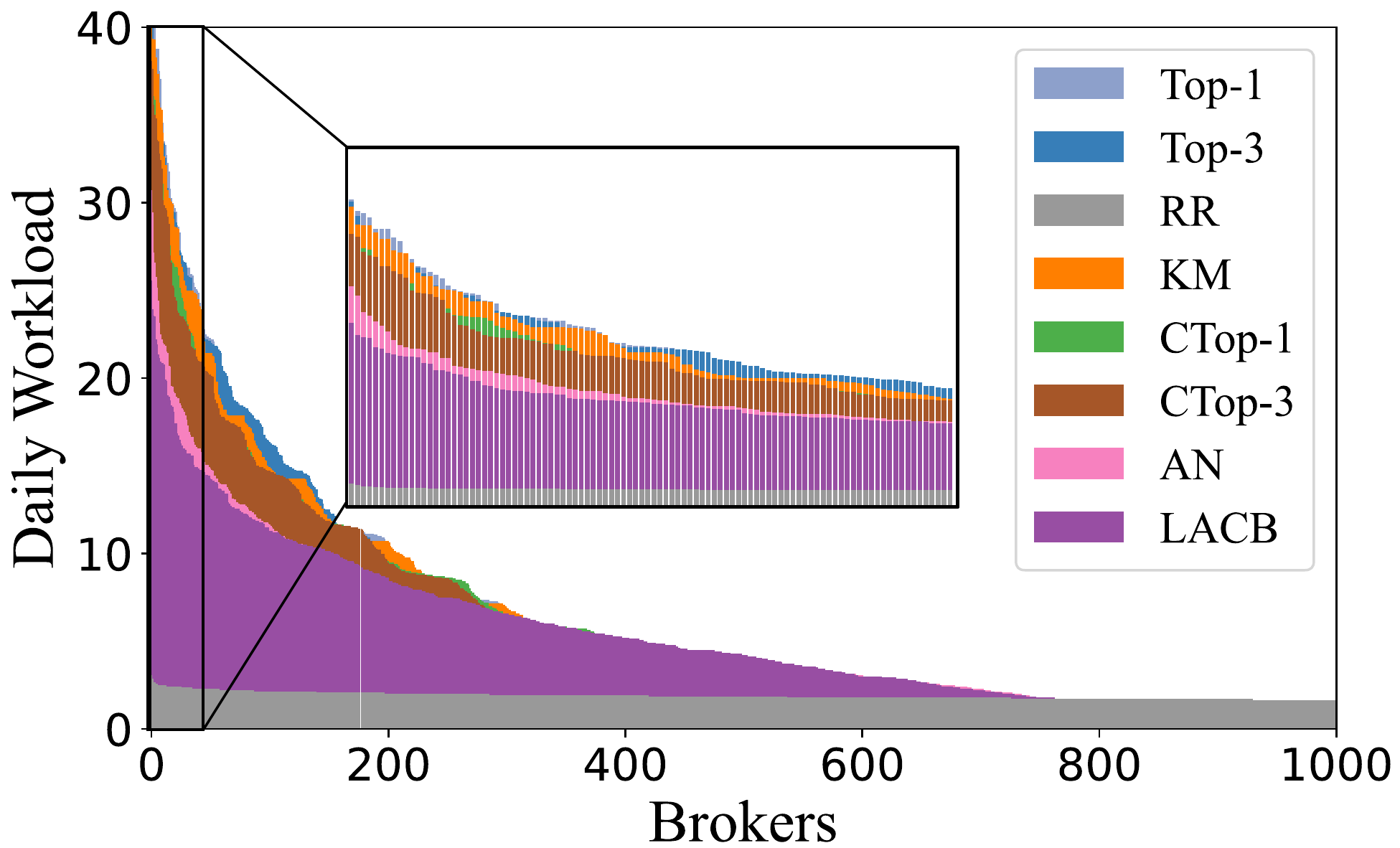}
        }
        \centering
        \vspace{-0.5em}
        \caption{\small The workload distribution of all compared algorithms, where we can take Top-K as overloaded results.}
        \label{fig:exp_workload}
        \vspace{-1.8em}
    \end{figure*}

\fakeparagraph{Impact of \# requests} 
The second column of \figref{fig:syn_results} shows the results of varying $|R|$.
The total utility generally increases as $|R|$ increases.
Similar to the previous experiment, LACB and LACB-Opt achieve the same utility and perform better than other algorithms.
As for the running time, the KM-based algorithms are slower than others.
LACB-Opt is up to $16.4\times \sim 1091.9\times$ faster than KM, AL and LACB. 
LABT-Opt is also competitively fast compared with Top-K and CTop-K, especially with a small number of requests.

\fakeparagraph{Impact of \# covered days}
The third column in \figref{fig:syn_results} presents the results of varying the covered days $Day$.
Our methods still outperform other baselines in terms of total utility.
Particularly, AN yields less utility in covering seven days, indicating that it may face a cold start, while LACB consistently performs well when varying covering days. 
Again, LACB and LACB-Opt perform the same in terms of utility.
As for the running time, a similar trend is observed as varying $|B|$.
Our LACB-Opt is $65.5\times \sim  329.4\times$ faster than other KM-based algorithms.

\fakeparagraph{Impact of the degree of imbalance}
The fourth column in \figref{fig:syn_results} plots the different imbalance ratios $\sigma$.
Since we fix $|B|$ and change $|R|$ to test different imbalance ratios, all algorithms have similar trends in utility as $\sigma$ increases.
The acceleration of LACB-Opt over LACB is closely related to the imbalance ratio $\sigma$.
For example, LACB-Opt is $641.7\times$ and $16.4\times$ faster than LACB when $\sigma=0.005$ and $\sigma=0.05$, respectively.

\subsection{Performance on Real Datasets}

\begin{figure}[h]
    \centering
    \subcaptionbox{Utility of City A}{
    \centering
    \includegraphics[width=0.45\linewidth]{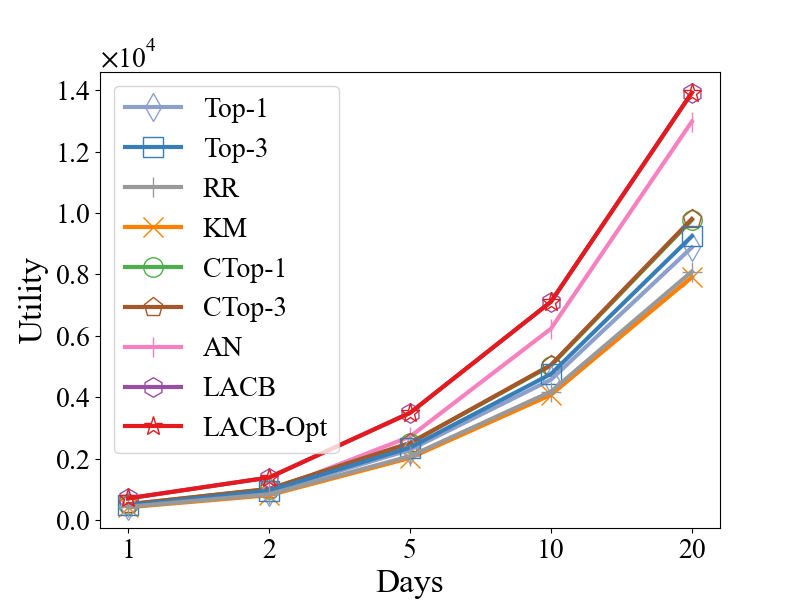}
    }
    \subcaptionbox{Time of City A}{
    \centering
    \includegraphics[width=0.45\linewidth]{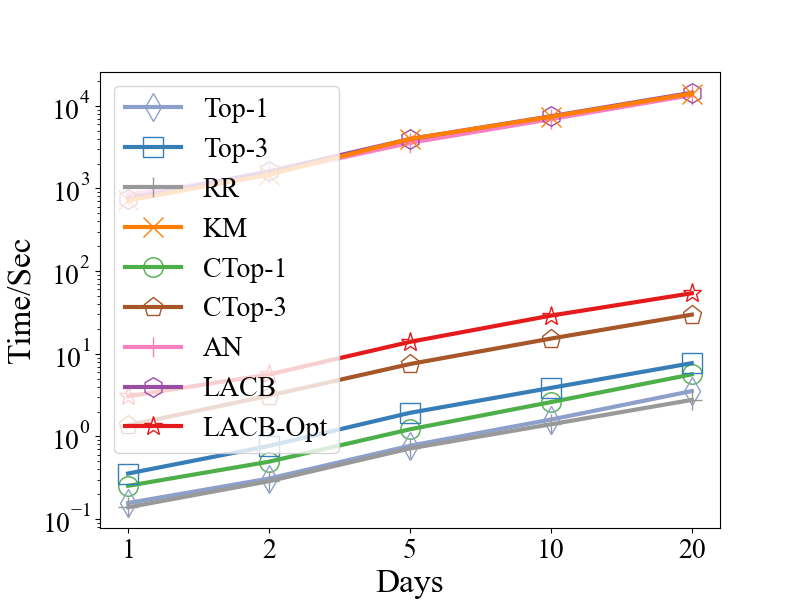}
    }
    
    \subcaptionbox{Utility of City B}{
    \centering
    \includegraphics[width=0.45\linewidth]{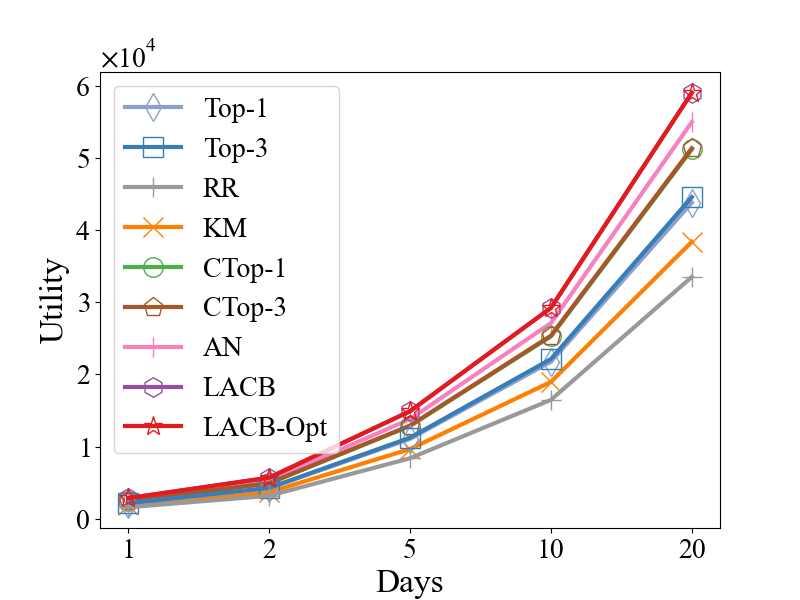}
    }
    \centering
    \subcaptionbox{Time of City B}{
    \centering
    \includegraphics[width=0.45\linewidth]{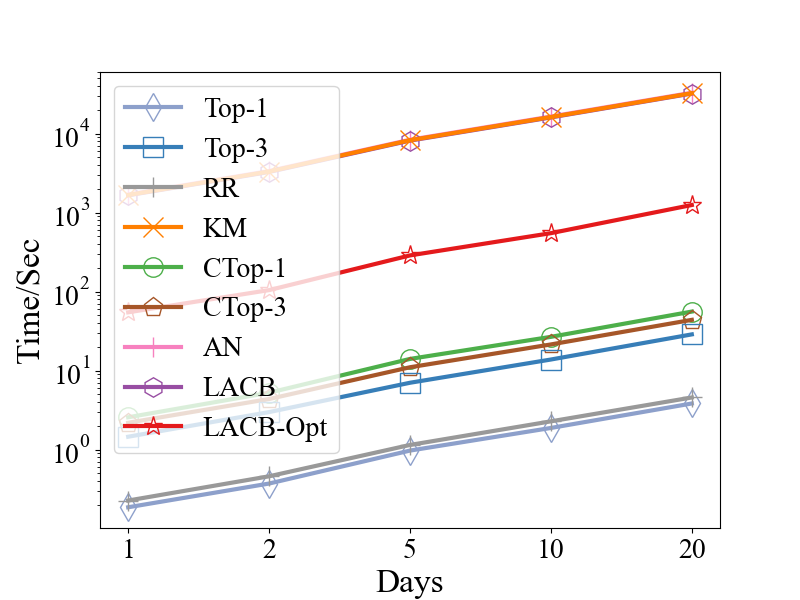}
    }
    \subcaptionbox{Utility of City C}{
    \centering
    \includegraphics[width=0.45\linewidth]{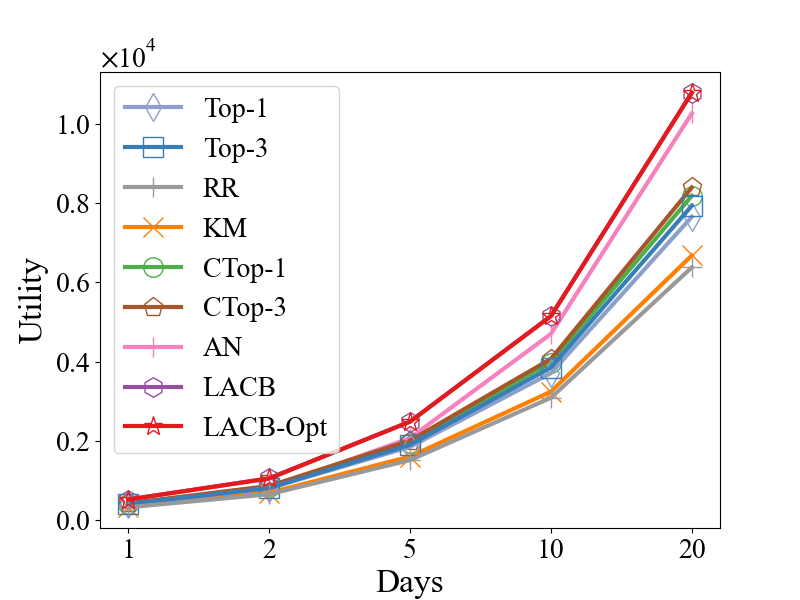}
    }
    \subcaptionbox{Time of City C}{
    \centering
    \includegraphics[width=0.45\linewidth]{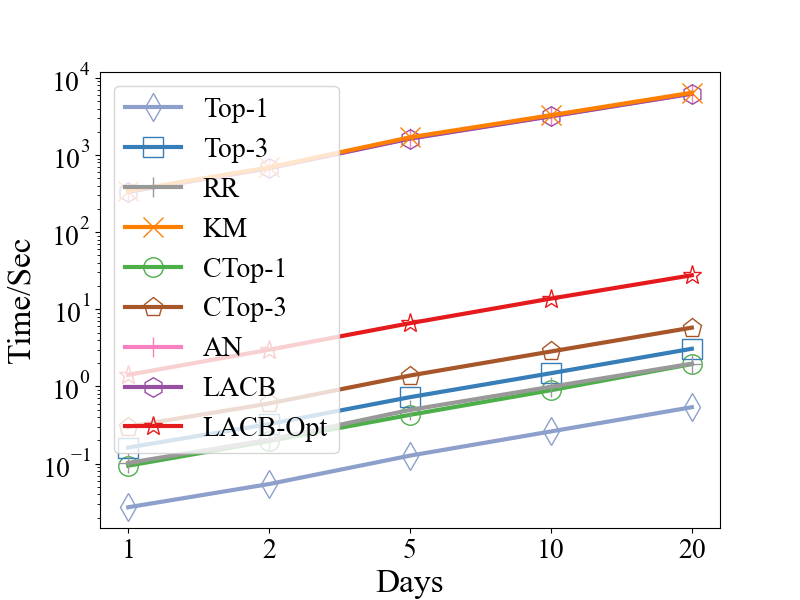}
    }
    \caption{
        \small Results on real-world datasets.
    }
    \label{fig:real_results}    
    \vspace{-1em}
\end{figure}

\fakeparagraph{Overall Performance}
\figref{fig:real_results} presents the results on three real-world datasets.
Take City A as an example and we can first make the following observations in total utility.
As expected, Top-K performs poorly on all three datasets.
Top-3 slightly outperforms Top-1 because Top-1 more easily overloads the recommended brokers.
CTop-K improves the total utility over Top-K, indicating the necessity of capacity awareness.
AN outperforms most baselines due to contextual bandits while our LACB and LACB-Opt outperform AN.
The running time increases linearly over days.
Similar to the results on synthetic datasets, KM, AN, and LACB are the slowest due to their cubic time complexity.
LACB-Opt is competitive in efficiency, only {$1.7 \sim 24.2$} seconds slower than Top-K and CTop-K and $233.4 \times \sim 284.9\times$ faster than other KM-based algorithms.
We have similar observations on data from all three cities.  

\fakeparagraph{In-depth Analysis of Brokers}
To understand the gain of our algorithms over baselines, we take a closer look at the distributions of both utility and workload of baselines.
As aforementioned, LACB-Opt only differs from LACB in efficiency, so there is no need to include LACB-Opt in the following analysis.
\shuyue{As we focus on the utility/workload of top brokers, we only show brokers with higher utility or workload.
The utility/workload of other brokers exhibits a similar long-tail distribution as \figref{fig:exp_utility_dis} and \figref{fig:exp_workload}, and is omitted here.}
\begin{itemize}
    \item \textit{Utility Distribution}.
    \figref{fig:exp_utility_dis} plots the utility distribution of all algorithms.
    Take City A as an example.
    Capacity-based assignment algorithms (CTop-K, AN, and LACB) achieve higher utility than Top-K for most brokers, \ie avoiding the overloaded phenomenon is also beneficial at individual levels.
    Particularly, 80.8\% brokers in LACB have an improvement in utility compared with Top-K.
    RR results in closer utilities for most brokers, as it randomly apportions requests to all brokers.
    However, RR decreases the utility of 25.7\% brokers compared with Top-K.
    Similar observations are found in City B and City C.
    To summarize, LACB can improve the utility of both top brokers and the remaining brokers.
    \item \textit{Workload Distribution}.
    \figref{fig:exp_workload} plots the workload distribution of all algorithms.
    As expected, Top-K leads to the highest workload of top brokers.
    RR randomly apportions requests and results in the lowest workload of top brokers.
    Yet it also prevents top brokers from serving more even if they have spare capacity, thus limiting the potential of top brokers.
    Except for RR, the workloads of top brokers in LACB are the lowest, which means that top brokers in LACB are at low risk of overload.
    To summarize, LACB outperforms baselines thanks to its suitable workload capacity for top brokers.
\end{itemize}

\subsection{Summary of Experimental Results}
We summarize our major experimental findings as follows.
\begin{itemize}
    \item 
    CTop-K outperforms Top-K on all datasets, indicating the necessity of setting a workload capacity.
    \item 
    Our solutions, LACB and LACB-Opt, generally outperform other baselines in total utility.
    They also improve 72.0\%$\sim $82.2\% brokers' utility compared with Top-K.
    \item 
    Without loss of utility, LACB-Opt is up to {$284.9\times$} faster than other KM-based algorithms on real-world datasets, which is consistent with our theoretical analysis.
\end{itemize}


\section{Related Work} \label{sec:related}
Our work is mainly related to three lines of research:  \textit{online task assignment}, \textit{data science for real estate}, and \textit{contextual bandits}.
We review the representative work as follows.

\fakeparagraph{Online Task Assignment} 
Task assignment is the core operation of crowdsourcing \cite{VLDB20_Survey, ICDE18_task_assignment, ICDE17_survey, VLDB16_Crowd, VLDBJ15_Crowd_TA, TKDE21_SC,KDD18_Du,SIGMOD17_Tong,TongJoS17}.
We focus on the online task assignment and summarize previous work into two categories: the vertex-based mode and the batch-based mode.
\textit{(i) vertex-based mode}. 
In \cite{GIS12}, Kazemi \etal study the bipartite matching in spatial crowdsourcing and adopt greedy-based assignment algorithms to optimize total utility.
Tong \etal \cite{VLDB16_Exp} conduct the experimental study on online bipartite matching and show that the greedy algorithm is competitive in many practical settings.
Later studies explore task assignment with different objectives or constraints, such as fairness \cite{TOSC18_fair, KDD21_fair, ICDE21_fairness}, privacy \cite{ICDE20_privacy, ICDE18_privacy, NeuroCom} and incentive \cite{SIGMOD18_bandit_price, MDM21_incentive,WAIM16_SC}. 
Particularly, the fair-aware task assignment algorithm in \cite{TOSC18_fair} can be extended as a baseline to avoid the overloaded phenomenon in our scenario.
\textit{(ii) batch-based mode}. 
This mode prevails in ride-hailing \cite{KDD17_dispatch, ICDE19_batch, KDD18_batch}.
    Wang \etal \cite{ICDE19_batch} adopt reinforcement learning to adaptively adjust the time window of each batch in the assignment.
    Zhang \etal \cite{KDD17_dispatch} solve the batched task assignment by a heuristic hill-climbing method.
    Authors in \cite{KDD18_batch} and \cite{ICDE22_Adaptive} take the value function to optimize the sequential decision-making over a long horizon, which inspires one of the core ideas of our method.
Existing solutions mainly focus on the assignment strategies and assume a known capacity, however,
our scenario is more challenging as we need to estimate the personalized capacity before assignments.

\fakeparagraph{Data Science For Real Estate}
The availability of big data has stimulated the interest in adopting data-driven approaches in the real estate industry \cite{KDD20_tutorial}.
Prior work mainly focus on two topics, housing appraisal \cite{KDD14_house} and housing finding  \cite{KDD18_house_rec}.
In \cite{KDD14_house}, Fu \etal enhance real estate appraisal by modeling dependencies of nearby estates and affiliated business areas.
Zhang \etal \cite{KDD21_house} propose the graph neural networks based method for the asynchronously spatio-temporal dependencies.
Grbovic \etal \cite{KDD18_house_rec} design embedding techniques for real-time personalized housing searching.
In \cite{CHI18_house}, Weng \etal improve the home-finding service by analyzing the reachability between the home locations and concerned POIs.

In this work, we investigate broker matching, an essential issue for online real estate platforms yet has rarely been optimized from a data-driven perspective.

\fakeparagraph{Contextual Bandits}
The contextual bandit utilizes additional information to make decisions, which has been adopted in recommender system \cite{WWW10_bandit_rec}, query optimization \cite{SIGMOD21_bandit} and dynamic pricing \cite{SIGMOD18_bandit_price}.
Early efforts mainly focus on the theoretical analysis of regret bound \cite{CEC20_bandit_survey, NIPS07_bandit_th}.
Li \etal \cite{WWW10_bandit_rec} propose the LinUCB algorithm, which combines a linear predictor with upper bound confidence algorithms in 2010 and follow-up researchers extend the basic framework of LinUCB for efficient learning \cite{KDD17_bandit} and latent relations representation \cite{CIKM16_bandit}.
Recently, some work utilize neural networks to optimize the contextual bandit \cite{ICML20_nn_bandit, SIGMOD21_bandit, KDD21_nn_bandit}.
Zhou \etal \cite{ICML20_nn_bandit} design a neural contextual bandit algorithm without assumption about the reward function, which can be extended as a baseline in our scenario.
Marcus \etal \cite{SIGMOD21_bandit} propose contextual bandits with tree convolutional neural networks to learn query execution strategies in databases.
Ban \etal \cite{KDD21_nn_bandit} propose multi-facet contextual bandits, where each facet is designed to characterize one of users’ needs.
However, these algorithms are not directly applicable in case of personalized estimation, as is in our capacity-aware broker matching scenario.
Our LACB is built upon this thread of research but improves prior studies by adopting layer transfer for more data-efficient bandit learning.

\section{Conclusion} \label{sec:conclusion}
In this paper, we investigate the overloaded phenomenon in broker matching for online real estate platforms.
Specifically, we find that the top-k recommendation mechanism adopted by mainstream platforms for broker matching tends to overwhelm top brokers, which decreases both sign-up rates and total utility.
The root cause lies in the ignorance of the broker's workload capacity.
In response, we propose {LACB}, which effectively estimates the personalized capacity by transferable neural contextual bandits and assigns brokers to clients from a global view.
{We further propose LACB-Opt, which largely improves the efficiency of {LACB} on an imbalanced bipartite graph.}
Extensive evaluations on three cities based on a real-world online real estate platform demonstrate the effectiveness {and efficiency} of our approach.
We envision our work as an insightful reference for a wide spectrum of service industries where workers have limited workload capacity.

\bibliographystyle{IEEEtran}
\bibliography{reference.bib}

\end{document}